\documentclass[12pt,reqno]{amsart}

\usepackage[numbers,sort&compress,square]{natbib}
\usepackage{amsmath}
\usepackage{amsthm}
\usepackage{amssymb}
\usepackage{eucal}
\usepackage{mathrsfs}
\usepackage{fullpage}
\usepackage{setspace}
\usepackage{bm}
\usepackage{color}
\usepackage{mathabx}
\usepackage{graphicx}
\usepackage{lineno}
\usepackage{dsfont}
\usepackage{amsaddr}
\usepackage[font=scriptsize]{caption}

\title{Fixation probabilities for any configuration of two strategies on regular graphs}

\author{Yu-Ting Chen$^{1,2,3}$, Alex McAvoy$^{1,4}$, and Martin A. Nowak$^{1,5,6}$}

\address{\small$^{1}$Program for Evolutionary Dynamics, Harvard University, Cambridge, MA 02138 \\ $^{2}$Center of Mathematical Sciences and Applications, Harvard University, Cambridge, MA 02138 \\ $^{3}$Department of Mathematics, University of Tennessee, Knoxville, TN 37996 \\$^{4}$Department of Mathematics, University of British Columbia, 1984 Mathematics Road, Vancouver, BC, Canada V6T 1Z2 \\ $^{5}$Department of Mathematics, Harvard University, Cambridge, MA 02138 \\ $^{6}$Department of Organismic and Evolutionary Biology, Harvard University, Cambridge, MA 02138}

\theoremstyle{definition}

\newtheorem{lemma}{Lemma}
\newtheorem{proposition}{Proposition}
\newtheorem{remark}{Remark}
\newtheorem{theorem}{Theorem}

\newcommand{\T}{\intercal}
\newcommand{\eq}[1]{Eq.~(\ref{eq:#1})}
\newcommand{\fig}[1]{Fig.~\ref{fig:#1}}
\newcommand{\lem}[1]{Lemma~\ref{lem:#1}}
\newcommand{\thm}[1]{Theorem~\ref{thm:#1}}

\begin{document}

\allowdisplaybreaks

\begin{abstract}
Population structure and spatial heterogeneity are integral components of evolutionary dynamics, in general, and of evolution of cooperation, in particular. Structure can promote the emergence of cooperation in some populations and suppress it in others. Here, we provide results for weak selection to favor cooperation on regular graphs for any configuration, meaning any arrangement of cooperators and defectors. Our results extend previous work on fixation probabilities of single, randomly placed mutants. We find that for any configuration cooperation is never favored for birth-death (BD) updating. In contrast, for death-birth (DB) updating, we derive a simple, computationally tractable formula for weak selection to favor cooperation when starting from any configuration containing any number of cooperators and defectors. This formula elucidates two important features: (i) the takeover of cooperation can be enhanced by the strategic placement of cooperators and (ii) adding more cooperators to a configuration can sometimes suppress the evolution of cooperation. These findings give a formal account for how selection acts on all transient states that appear in evolutionary trajectories. They also inform the strategic design of initial states in social networks to maximally promote cooperation. We also derive general results that characterize the interaction of any two strategies, not only cooperation and defection.
\end{abstract}

\maketitle

\section{Introduction}

Mechanisms favoring the emergence of cooperation in social dilemmas have become central focuses of evolutionary game theory in recent years \citep{nowak:Science:2006,brosnan:PTRSB:2010,zaggl:JIE:2013}. The dilemma of cooperation, which is characterized by conflicts of interest between individuals and groups, poses a significant challenge to models of evolution since many of these models predict that cooperation cannot persist in the presence of exploitation by defectors \citep{nowak:Nature:2004,nowak:BP:2006}. Yet cooperation is widely observed in nature, and the spatial assortment that results from population structure is one element that can promote its emergence. In fact, spatial structure is among the most salient determinants of the evolutionary dynamics of a population \citep{nowak:Nature:1992,lindgren:PD:1994,killingback:PRSB:1996,killingback:JTB:1998,brauchli:JTB:1999,szabo:PRE:2000,hutson:MAA:2002,szabo:PRL:2002,durrett:SNY:2002,ifti:JTB:2004,ranta:EP:2005,jansen:Nature:2006,komarova:BMB:2006,perc:PRE:2008,nowak:PTRSB:2009,roca:PLR:2009,fu:JTB:2010,helbing:NJP:2010,szolnoki:EPL:2010,vancleve:TPB:2013,schreiber:TPB:2013,hwang:TE:2013,kaveh:RSOS:2015}.

In social dilemmas, population structure can allow for the emergence of localized cooperative clusters that would normally be outcompeted by defectors in well-mixed populations \citep{nowak:Nature:2004,rand:PNAS:2014}. However, whether population structure promotes or suppresses cooperation depends on a number of factors such as the update rule, the type of social dilemma, and the spatial details of the structure (which determines the extent of local competition \citep[see][]{vancleve:TPB:2013,debarre:NC:2014}). For example, cooperation need not be favored in prisoner's dilemma interactions under all update rules \citep{ohtsuki:Nature:2006,ohtsuki:JTB:2006,tarnita:JTB:2009,allen:EMS:2014}. As a consequence, population structure should be considered in the context of the game and the underlying update rule.

In the donation game, a cooperator ($C$) pays a cost, $c$, to provide the opponent with a benefit, $b$, and a defector ($D$) pays no cost and provides no benefit  \citep{sigmund:PUP:2010}. Provided $b>c>0$, this game represents a prisoner's dilemma since then the unique Nash equilibrium is mutual defection, but both players would prefer the payoff from mutual cooperation \citep{maynardsmith:CUP:1982}. In addition to representing one of the most important social dilemmas, the donation game also admits a simple way in which to quantify the efficiency of cooperation: the benefit-to-cost ratio, $b/c$. As this ratio gets larger, the act of cooperation has a more profound effect on the opponent relative to the cost paid by the cooperator. As we shall see, for any configuration of cooperators and defectors, this ratio is a vital indicator of the evolutionary performance of cooperation.

Evolutionary graph theory is a framework for studying evolution in structured populations \citep{lieberman:Nature:2005,ohtsuki:Nature:2006,ohtsuki:JTB:2006,szabo:PR:2007,taylor:Nature:2007,santos:Nature:2008,szolnoki:PRE:2009,broom:JSTP:2011,broom:JTB:2012,vanveelen:PNAS:2012,shakarian:B:2012,chen:AAP:2013,maciejewski:PLoSCB:2014,debarre:NC:2014}. In a graph-structured population, the players reside on the vertices and the edges indicate who is a neighbor of whom. In fact, there are two types of neighborhoods: (i) those that generate payoffs (``interaction neighborhoods") and (ii) those that are relevant for evolutionary updating (``dispersal neighborhoods"). Thus, an evolutionary graph is actually a pair of graphs consisting of an interaction graph and a dispersal graph \citep{ohtsuki:PRL:2007,taylor:Nature:2007,ohtsuki:JTB:2007,pacheco:PLoSCB:2009}. As in many other studies, we assume that the interaction and dispersal graphs are the same. Other extensions of evolutionary graph theory involve dynamic graphs, which allow the population structure to change during evolutionary updating \citep{antal:PNAS:2009,tarnita:PNAS:2009,wu:PLoSONE:2010,wardil:SR:2014}. Our focus is on static, regular graphs of degree $k$, meaning the population size, $N$, is fixed and each player has exactly $k$ neighbors.

We study two prominent update rules: birth-death (BD) and death-birth (DB). In both processes, players are arranged on a graph and accumulate payoffs by interacting with all of their neighbors. This payoff, $\pi$, is then converted to fitness, $f$, via $f=1+w\pi$, where $w\geqslant 0$ is the intensity of selection \citep{nowak:Nature:2004}. For BD updating  \citep{moran:MPCPS:1958,nowak:Nature:2004}, a player is chosen with probability proportional to fitness for reproduction; the offspring of this player then replaces a random neighbor (who dies). For DB updating \citep{ohtsuki:Nature:2006}, a player is chosen uniformly at random for death; a neighbor of this player then reproduces (with probability proportional to fitness) and the offspring fills the vacancy. For each of these processes, we assume that $w$ is small, which means selection is weak. Weak selection is often a biologically meaningful assumption since an individual might possess many traits (strategies), and each trait makes only a small contribution to fitness \citep{nowak:Nature:2004,wild:JTB:2007,fu:PRE:2009,wu:PRE:2010,akashi:G:2012,wu:PLoSCB:2013,mullon:JEB:2014}.

The effects of selection on fixation probability have been studied chiefly for states with just a single cooperator since, if the mutation rate is small, the process will reach a monomorphic state prior to the appearance of another cooperator through mutation \citep{fudenberg:JET:2006,wu:JMB:2011}. Although small mutation rates are often reasonable from a biological standpoint \citep{otto:ME:1998,bromham:BL:2009,loewe:PTRSB:2010,lynch:TG:2010}, there are several reasons to study arbitrary cooperator configurations. Even when starting from a state with a single cooperator, an evolutionary process typically transitions subsequently through states with many cooperators. From a mathematical standpoint, it is therefore natural to ask how selection affects the fixation probability of cooperators from each possible transient state that might arise in an evolutionary trajectory. Furthermore, many-mutant states could arise through migration \citep{miekisz:LNM:2008,ohtsuki:E:2010,hauert:JTB:2012,pichugin:JTB:2015} or environmental mutagenic agents \citep{nagao:ARG:1978,nunney:EA:2012}, which, even when rare, might result in several cooperators entering the population at once. In the case of social networks, cooperators could arise through design rather than mutation or exploration; if cooperators can be strategically planted within the population, then one can ask how to do so in order to maximize the chances that cooperators take over. Therefore, the effects of selection on arbitrary numbers and configurations of cooperators and defectors play an important role in the evolutionary dynamics of cooperation.

When starting from a configuration with $n$ cooperators and $N-n$ defectors, weak selection is said to favor the evolution of cooperation (on a regular graph) if the probability that cooperators fixate exceeds $n/N$ if $w$ is sufficiently small but positive. This comparison is based on the fact that the fixation probability of $n$ cooperators for neutral drift ($w=0$) is $n/N$. \citet{ohtsuki:Nature:2006} show that, on large regular graphs of degree $k$, selection favors the fixation of a single, randomly-placed cooperator under DB updating as long as
\begin{linenomath}
\begin{align}
\frac{b}{c} > k .
\end{align}
\end{linenomath}
\citet{taylor:Nature:2007} show that for finite bi-transitive graphs of size $N$ and degree $k$, the condition for selection to favor the fixation of a single cooperator is
\begin{linenomath}
\begin{align}\label{eq:ratioRandom}
\frac{b}{c} > \frac{k\left(N-2\right)}{N-2k} .
\end{align}
\end{linenomath}
Bi-transitive graphs constitute a subset of regular graphs.

In another refinement of the `$b/c>k$' result, \citet{chen:AAP:2013} shows that, for any $n$ with $0<n<N$, selection favors cooperation when starting from a random configuration of $n$ cooperators and $N-n$ defectors on a regular graph of size $N$ and degree $k$ if and only if Eq. (\ref{eq:ratioRandom}) holds. (Note that regularity is a weaker requirement on the population structure than bi-transitivity.) This ratio, which characterizes when selection increases the fixation probability of cooperators, is independent of the location of the mutants, despite the fact that the probability of fixation itself depends on the location \citep{mcavoy:JRSI:2015}. As the population size, $N$, gets large, the critical benefit-to-cost ratio of \eq{ratioRandom} approaches $k$, which recovers the result of \citet{ohtsuki:Nature:2006}. Our goal here is to move beyond \eq{ratioRandom} and give an explicit, computationally feasible critical benefit-to-cost ratio for any configuration of cooperators and defectors on any regular graph.

Given the profusion of possible ways to structure a population of a fixed size, it quickly becomes difficult to determine when a population structure favors the evolution of cooperation. Here, we provide a solution to this problem for BD and DB updating on regular graphs. We show that, for any configuration of cooperators and defectors, (i) cooperation is never favored for BD updating, and (ii) for DB updating, there exists a simple, explicit critical benefit-to-cost ratio that characterizes when selection favors the emergence of cooperation. Moreover, if $N$ is the population size and $k$ is the degree of the graph, then the complexity of calculating this ratio is $O\left(k^{2}N\right)$, and, in particular, linear in $N$. Thus, while the calculations of fixation probabilities in structured populations are famously intractable \citep{voorhees:PRSA:2013,ibsenjensen:PNAS:2015,hindersin:B:2016}, the determination of whether or not selection increases the probability of fixation, for weak selection, is markedly simpler.

In addition to providing a computationally feasible way of determining whether selection favors cooperation on a particular graph, our results highlight the importance of the initial configuration for the emergence of cooperation. Depending on the graph, adding additional cooperators to the initial condition can either suppress or promote the evolution of cooperation. A careful choice of configuration of cooperators and defectors can minimize the critical benefit-to-cost ratio for selection to favor cooperation. If cooperation is not favored by selection in such a strategically chosen initial state, then it cannot be favored under any other initial configuration. In this sense, there exists a configuration that is most conducive to the evolution of cooperation, which is not apparent from looking at single-cooperator configurations or random configurations with $n$ cooperators since these initial configurations need not minimize the critical benefit-to-cost ratio.

\section{Results}

\subsection{Critical benefit-to-cost ratios}
Let $\xi$ be a configuration of cooperators and defectors on a fixed regular graph of size $N$ and degree $k$, and let $\mathbf{C}$ denote the configuration consisting solely of cooperators. For the donation game, the probability that cooperators take over the population when starting from state $\xi$ may be viewed as a function of the selection intensity, $\rho_{\xi ,\mathbf{C}}\left(w\right)$. We consider here the following question: when does weak selection increase the probability that cooperators fixate? In other words, when is $\rho_{\xi,\mathbf{C}}\left(w\right) >\rho_{\xi,\mathbf{C}}\left(0\right)$ for sufficiently small $w>0$? Note that if there are $n$ cooperators in state $\xi$, then $\rho_{\xi ,\mathbf{C}}\left(0\right) =n/N$, so this condition is equivalent to $\rho_{\xi ,\mathbf{C}}\left(w\right) >n/N$ for small $w>0$.

To answer this question, we first need to introduce some notation. If $x$ is a vertex of the graph and $\xi$ is a configuration, then let $f_{1}\left(x,\xi\right)$ and $f_{0}\left(x,\xi\right)$ be the frequencies of cooperators and defectors, respectively, among the neighbors of the player at vertex $x$. Similarly, let $f_{10}\left(x,\xi\right)$ be the fraction of paths of length two, starting at $x$, that consist of a cooperator followed by a defector. From these quantities, let
\begin{linenomath}
\begin{subequations}\label{eq:localFrequencies}
\begin{align}
\overline{f_{1}} &:= \frac{1}{N}\sum_{x\in V}f_{1}\left(x,\xi\right) ; \\
\overline{f_{0}} &:= \frac{1}{N}\sum_{x\in V}f_{0}\left(x,\xi\right) ; \\
\overline{f_{10}} &:= \frac{1}{N}\sum_{x\in V}f_{10}\left(x,\xi\right) ; \\
\overline{f_{1}f_{0}} &:= \frac{1}{N}\sum_{x\in V}f_{1}\left(x,\xi\right) f_{0}\left(x,\xi\right) ,
\end{align}
\end{subequations}
\end{linenomath}
which are obtained by averaging these `local frequencies' over all of the players in the population. From these local frequencies, which are straightforward to calculate (see \fig{localFrequencies}), we obtain our main result: for small $w>0$, $\rho_{\xi,\mathbf{C}}\left(w\right) >\rho_{\xi,\mathbf{C}}\left(0\right)$ if and only if the benefit-to-cost ratio exceeds the critical value
\begin{linenomath}
\begin{align}\label{eq:ratioGeneral}
\left(\frac{b}{c}\right)_{\xi}^{\ast} &= \frac{k\left(N\overline{f_{1}}\cdot\overline{f_{0}}-\overline{f_{10}}\right)}{N\overline{f_{1}}\cdot\overline{f_{0}}-k\overline{f_{10}}-k\overline{f_{1}f_{0}}}
\end{align}
\end{linenomath}
whenever the denominator is positive (and $\infty$ otherwise). Since the calculations of $\overline{f_{1}}$, $\overline{f_{0}}$, and $\overline{f_{1}f_{0}}$ are $O\left(kN\right)$ and the calculation of $\overline{f_{10}}$ is $O\left(k^{2}N\right)$, it follows that the complexity of finding the critical benefit-to-cost ratio is $O\left(k^{2}N\right)$, so it is feasible to calculate even when the population is large. Note also that if $\widehat{\xi}$ is the state obtained by swapping cooperators and defectors in $\xi$, then both $\xi$ and $\widehat{\xi}$ have the same critical benefit-to-cost ratio. We discuss these `conjugate' states further in our treatment of structure coefficients.

\begin{figure}
\centering
\includegraphics[scale=2.0]{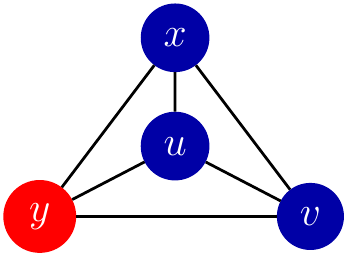}
\caption{Calculation of the local frequencies of \eq{localFrequencies}, $f_{1}\left(x,\xi\right)$, $f_{0}\left(x,\xi\right)$, and $f_{10}\left(x,\xi\right)$, where $\xi$ is the configuration consisting of a defector at vertex $y$ and cooperators elsewhere. Among the three neighbors of the player at vertex $x$, two are cooperators ($u$ and $v$) and one is a defector ($y$); thus, $f_{1}\left(x,\xi\right) =2/3$ and $f_{0}\left(x,\xi\right) =1/3$. Furthermore, of the nine paths of length two that begin at vertex $x$, only two ($x\rightarrow u\rightarrow y$ and $x\rightarrow v\rightarrow y$) consist of a cooperator followed by a defector, and it follows that $f_{10}\left(x,\xi\right) =2/9$.\label{fig:localFrequencies}}
\end{figure}

When $\xi$ has just a single cooperator, the ratio of \eq{generalPayoffMatrix} reduces to that of \eq{ratioRandom}, which, in particular, does not depend on the location of the cooperator. This property is notable because the fixation probability itself usually does depend on the location of the cooperator, even on regular graphs \citep{mcavoy:JRSI:2015}. We show in Methods that one recovers from \eq{ratioGeneral} the result of \citet{chen:AAP:2013} that \eq{ratioRandom} gives the critical benefit-to-cost ratio for a randomly-chosen configuration with a fixed number of cooperators.

For fixed $k\geqslant 2$, the critical benefit-to-cost ratio in \eq{ratioGeneral} converges uniformly to $k$ as $N\rightarrow\infty$ (see Methods). Therefore, on sufficiently large graphs, the critical ratio is approximated by $k$ for any configuration, regardless of the number of cooperators. As a result, on large graphs there is less of a distinction between the various transient (non-monomorphic) states in terms of whether or not selection favors the fixation of cooperators. On smaller graphs, these transient states can behave quite differently from one another. This effect is particularly pronounced on very small social networks in which cooperators can be strategically planted in the population to ensure that cooperators are favored by selection.

\subsubsection{Strategic placement of cooperators in (small) social networks}

Among the more interesting consequences of \eq{ratioGeneral} are its implications for the success of cooperators as a function of the initial configuration. Recall that \eq{ratioRandom} gives the critical benefit-to-cost ratio for both (i) configurations with a single cooperator and (ii) random configurations with a fixed number of cooperators. When cooperators and defectors are configured randomly, this critical ratio is independent of the number of cooperators, which suggests that the effects of selection cannot be improved by increasing the initial abundance of cooperators.

\eq{ratioGeneral}, on the other hand, shows that the initial configuration of cooperators, including their abundance, does affect how selection acts on the population. First of all, there are graphs for which the critical benefit-to-cost ratio is infinite for configurations with a single cooperator but finite for some configurations with multiple cooperators (see \fig{promoteSuppress}(a)). In contrast, there are graphs for which this ratio is finite for configurations with a single cooperator but infinite for some states with multiple cooperators (see \fig{promoteSuppress}(b)). Therefore, despite the fact that the critical ratio for a single mutant is the same as the critical ratio for a random configuration with any fixed number of mutants, the critical ratio does (in general) depend on the number of mutants present in the configuration.

\begin{figure}
\centering
\includegraphics[scale=1.0]{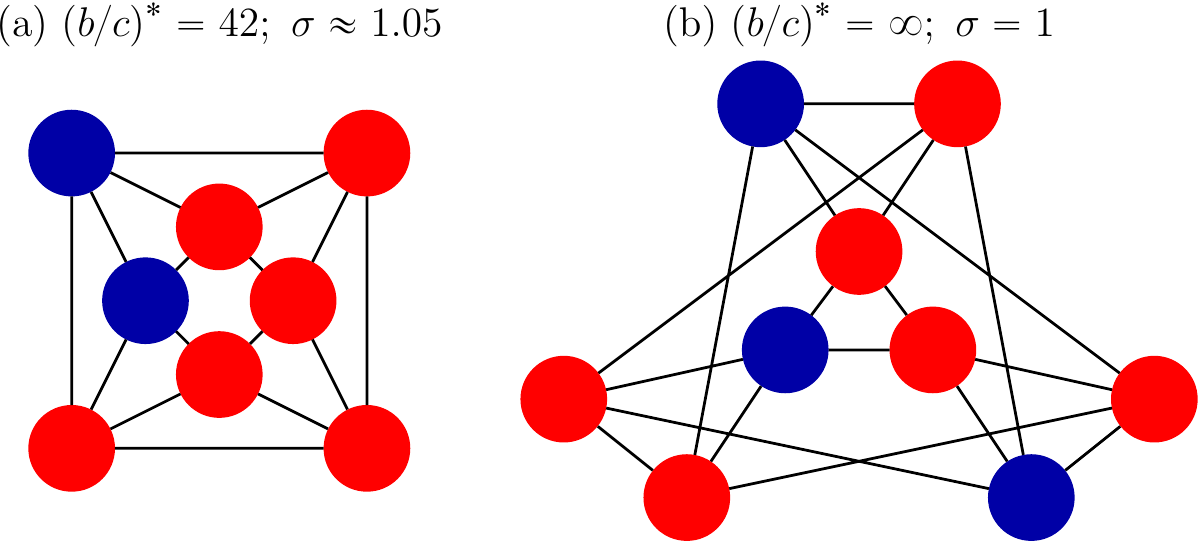}
\caption{Two graphs showing configurations of cooperators (blue) and defectors (red). (a) Cooperation can be favored for the initial condition that is shown since the critical benefit-to-cost ratio is $42$ and, in particular, finite. However, the fixation of cooperation cannot be favored for any initial configuration with a single cooperator on this graph. (b) Cooperation cannot be favored for the initial condition that is shown since the critical benefit-to-cost ratio is infinite. However, any initial configuration with a single cooperator has a critical benefit-to-cost ratio of $28$. Therefore, the addition of cooperators to the initial configuration can either favor cooperation, (a), or suppress it, (b). The critical benefit-to-cost ratio can also be expressed in terms of a well-known quantity known as a ``structure coefficient," $\sigma$, which satisfies $\left(b/c\right)^{\ast}=\left(\sigma +1\right) /\left(\sigma -1\right)$.\label{fig:promoteSuppress}}
\end{figure}

We say that a configuration has isolated cooperators (resp. defectors) if the minimum distance between any two cooperators (resp. defectors) is at least three steps. Let $N_{0}$ denote the maximum number of isolated strategies that a configuration can carry. (Examples of configurations with isolated cooperators on a graph with $N_{0}=3$ are given in \fig{fruchtMinMax}.) If a strategy (cooperate or defect) appears only once in a configuration, then that strategy is clearly isolated, so $N_{0}\geqslant 1$. We show in Methods that if $N>2k$, then cooperation can be favored for a mixed initial condition with $n$ cooperators whenever $1\leqslant n\leqslant N_{0}+1$ or $1\leqslant N-n\leqslant N_{0}+1$, and, moreover, these bounds on $n$ are sharp. Stated differently, under these conditions any configuration with $n$ cooperators has a finite critical benefit-to-cost ratio. Furthermore, if $N_{0}\geqslant 2$, then, for any $n$ with $2\leqslant n\leqslant N_{0}$, there exists a configuration with $n$ cooperators whose critical ratio is strictly less than the ratio for a single cooperator (\eq{ratioRandom}). Such a configuration necessarily has no isolated strategies since the minimum critical ratio among configurations with an isolated strategy is attained by any state with just a single cooperator.

\begin{figure}
\centering
\includegraphics[scale=1.0]{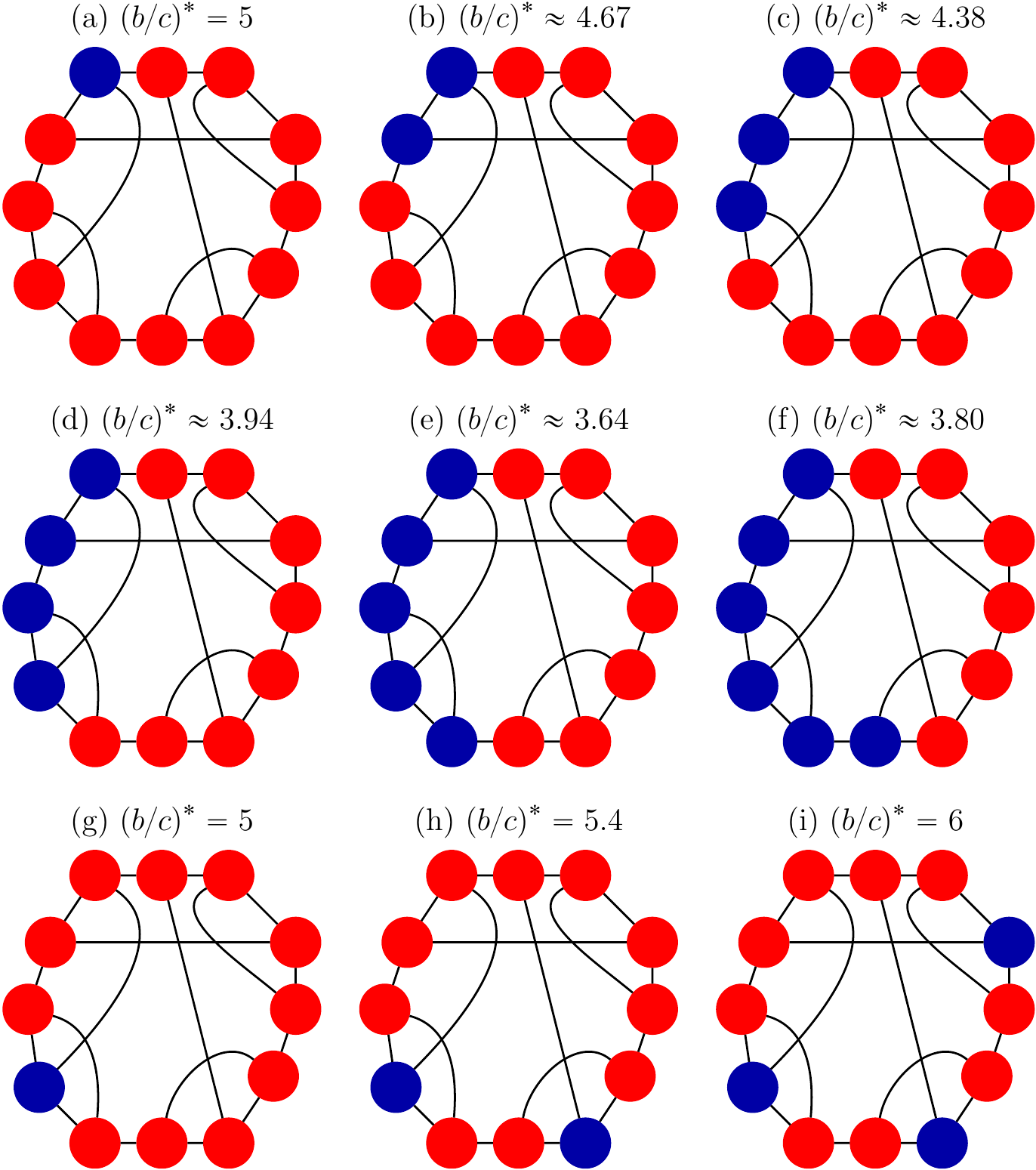}
\caption{Configurations of cooperators and defectors on the Frucht graph, a $3$-regular graph with $12$ vertices and no non-trivial symmetries \citep[see][]{frucht:CJM:1949}. Panels (a)-(f) show the effects on the critical benefit-to-cost ratio of adding additional cooperators to the initial state. Panel (e) shows the global minimum of $\left(b/c\right)_{\xi}^{\ast}$, which is achieved by just (e) and its conjugate; adding additional cooperators to the configuration in (e) only increases $\left(b/c\right)_{\xi}^{\ast}$. The configuration of (e) is `optimal' for cooperation in the sense that if selection increases the fixation probability of cooperators in some state, then it does so in state (e) as well. Relative to all possible initial states, selection can increase the fixation probability of cooperators in (e) under the smallest $b/c$ ratio. Panels (g)-(i) show that when cooperators are added in a different order (starting with just a single cooperator), the critical benefit-to-cost ratio can actually be increased. Each of these three configurations has isolated cooperators, and (i) gives the global maximum of $\left(b/c\right)_{\xi}^{\ast}$, which is achieved by just (i) and its conjugate. Since $N_{0}=3$, (i) is a maximal isolated configuration. The initial state in (i) is least conducive to cooperation in the sense that, relative to all other initial configurations, (i) requires the largest $b/c$ ratio for selection to increase the fixation probability of cooperators. If selection increases this fixation probability when starting from state (i), then it does so when starting from any other mixed initial configuration.\label{fig:fruchtMinMax}}
\end{figure}

The strategic placement of cooperators and defectors can therefore produce a critical benefit-to-cost ratio that is less than the ratio for a single cooperator among defectors. In fact, starting from a configuration with just one cooperator, one can reduce this critical ratio by placing a second cooperator adjacent to the first cooperator (see Methods). If $b/c$ lies below \eq{ratioRandom} and above \eq{ratioGeneral}, then a strategically chosen configuration can ensure that the fixation of cooperation is favored by selection even if it is disfavored for any single-cooperator state. This behavior is particularly pronounced on small social networks, where the critical ratios take on a significant range of values (see \fig{fruchtMinMax}), and less apparent on large networks, where the critical ratios are much closer to the degree of the graph, $k$. Fortunately, on small networks it is easier to directly search for configurations that have small critical benefit-to-cost ratios via \eq{ratioGeneral}.

\subsection{Structure coefficients}
Consider now a generic $2\times 2$ game whose payoff matrix is
\begin{linenomath}
\begin{align}\label{eq:generalPayoffMatrix}
\bordermatrix{%
 & A & B \cr
A &\ a & \ b \cr
B &\ c & \ d \cr
}\ .
\end{align}
\end{linenomath}
The donation game is a special case of this game with $A$ indicating a cooperator and $B$ indicating a defector. If $\mathbf{A}$ denotes the monomorphic state consisting of only $A$-players and if $\xi$ is a configuration of $A$- and $B$-players, then a natural generalization of the question we asked for the donation game is the following: when is $\rho_{\xi ,\mathbf{A}}\left(w\right) >\rho_{\xi ,\mathbf{A}}\left(0\right)$ for sufficiently small $w>0$? That is, when does (weak) selection favor the fixation of $A$ when starting from state $\xi$? For technical reasons, this question is more difficult to answer when the payoff matrix is \eq{generalPayoffMatrix} instead of that of the donation game. There is, however, an alternative way of generalizing the critical benefit-to-cost ratio to \eq{generalPayoffMatrix}.

When considering the evolutionary success of strategy $A$ based on configurations with only one mutant, another standard measure is whether the fixation probability of a single $A$-mutant in a $B$-population exceeds that of a single $B$-mutant in an $A$-population \citep[see][Eq. \textbf{2}]{tarnita:JTB:2009}. That is, one compares the fixation probability of $A$ to the fixation probability of $B$ after swapping $A$ and $B$ in the initial state. This interchange of strategies may be defined for any initial state: formally, if $\xi$ is a configuration of $A$-players and $B$-players, the conjugate of $\xi$, written $\widehat{\xi}$, is the state obtained by swapping $A$ and $B$ in $\xi$. In other words, the $A$-players in $\xi$ are the $B$-players in $\widehat{\xi}$.

A natural generalization of this criterion to arbitrary configurations involves comparing the fixation probability of $A$ in $\xi$ to the fixation probability of $B$ in $\widehat{\xi}$. Let $\mathbf{A}$ and $\mathbf{B}$ be the monomorphic states consisting of all $A$-players and all $B$-players, respectively. In this context, our main result is that $\rho_{\xi ,\mathbf{A}}\left(w\right) >\rho_{\widehat{\xi},\mathbf{B}}\left(w\right)$ for all sufficiently small $w>0$ if and only if
\begin{linenomath}
\begin{align}\label{eq:sigmaRule}
\sigma_{\xi}a + b &> c + \sigma_{\xi}d ,
\end{align}
\end{linenomath}
where, for the DB updating,
\begin{linenomath}
\begin{align}\label{eq:structureCoefficient}
\sigma_{\xi} &= \frac{N\left(1+\frac{1}{k}\right)\overline{f_{1}}\cdot\overline{f_{0}} -2\overline{f_{10}}-\overline{f_{1}f_{0}}}{N\left(1-\frac{1}{k}\right)\overline{f_{1}}\cdot\overline{f_{0}}+\overline{f_{1}f_{0}}} .
\end{align}
\end{linenomath}
In Methods, we give an explicit formula for the structure coefficient, $\sigma_{\xi}$, for BD updating as well. Just as it is for the critical benefit-to-cost ratio of \eq{ratioGeneral}, the complexity of calculating $\sigma_{\xi}$ is $O\left(k^{2}N\right)$. In fact, the relationship between $\left(b/c\right)_{\xi}^{\ast}$ and $\sigma_{\xi}$ is remarkably straightforward:
\begin{linenomath}
\begin{align}\label{eq:ratioSigmaRelationship}
\left(\frac{b}{c}\right)_{\xi}^{\ast} &= \frac{\sigma_{\xi}+1}{\sigma_{\xi}-1} ,
\end{align}
\end{linenomath}
which, for DB updating, generalizes a result of \citet{tarnita:JTB:2009} to arbitrary configurations. Note that the critical benefit-to-cost ratio increases as $\sigma_{\xi}$ decreases. Moreover, unlike the critical benefit-to-cost ratio, $\sigma_{\xi}$ is always finite. Interestingly, both $\left(b/c\right)_{\xi}^{\ast}$ and $\sigma_{\xi}$ are invariant under conjugation, meaning they are the same for $\widehat{\xi}$ as they are for $\xi$.

For the donation game, \eq{sigmaRule} is equivalent to $b/c>\left(b/c\right)_{\xi}^{\ast}$. Of course, \eq{sigmaRule} applies to a broader class of games as well and represents a simple way to compare the success of a strategy ($A$) relative to its alternative ($B$) when selection is weak. In this sense, \eq{sigmaRule} may be thought of as a generalization of the critical benefit-to-cost rule to arbitrary $2\times 2$ games.

\section{Discussion}
Selection always opposes the emergence of cooperation for BD updating, regardless of the configuration of cooperators and defectors (see Methods). This result is consistent with previous studies showing that cooperation cannot be favored under random configurations \citep{ohtsuki:Nature:2006,ohtsuki:JTB:2006,grafen:JEB:2007,grafen:JTB:2008}, and it specifies further that cooperation cannot be favored under any configuration. For general $2\times 2$ games given by \eq{generalPayoffMatrix}, we show in Methods that one can also find a simple formula for $\sigma_{\xi}$ in the selection condition of \eq{sigmaRule} that can be easily calculated for a given graph.

Remarkably, for DB updating, both the critical benefit-to-cost ratio and $\sigma_{\xi}$ depend on only local properties of the configuration, which makes these quantities straightforward to calculate. Furthermore, the complexity of calculating both of these quantities is $O\left(k^{2}N\right)$, where $N$ is the size of the population and $k$ is the degree of the graph, so they are computationally feasible even on large graphs. Therefore, our results provide a tractable way of determining whether or not selection favors cooperation for any configuration.

Finding an optimal configuration, which is one that minimizes the critical benefit-to-cost ratio, seems to be a difficult nonlinear optimization problem. The critical ratio is easily computed for any given configuration, but a graph of size $N$ has $2^{N}$ possible configurations, which makes a brute-force search unfeasible for all but small $N$. Our results qualitatively show that both the abundance and the configuration of cooperators can strongly influence the effects of selection. We leave as an open problem whether it is possible to find a polynomial-time algorithm that produces an optimal configuration on any regular graph. However, since \eq{ratioGeneral} is extremely easy to compute for a given configuration, and since small graphs generally exhibit broader variations of critical ratios than do larger graphs (since $\left(b/c\right)_{\xi}^{\ast}\rightarrow k$ uniformly as $N\rightarrow\infty$), it is typically feasible to find a state that is more conducive to cooperation than a random configuration.

Our analysis of arbitrary configurations uncovers two important features of the process with DB updating: (i) there exist graphs that suppress the spread of cooperation when starting from a single mutant but promote the spread of cooperation when starting from configurations with multiple mutants (\fig{promoteSuppress}(a)), and (ii) there exist graphs that promote the spread of cooperation when starting from a single mutant but suppress the spread of cooperation when starting from configurations with many mutants (\fig{promoteSuppress}(b)). The proper initial configuration is thus a crucial determinant of the evolutionary dynamics, and our results help to engineer initial conditions that promote the emergence of cooperation on social networks. More importantly, these results provide deeper mathematical insights into the complicated problem of how selection affects the outcome of an evolutionary process at each point along an evolutionary trajectory.

\section{Methods}

\subsection{Notation and general setup}

In what follows, the population structure is given by a simple, connected, $k$-regular graph, $G=\left(V,E\right)$, where $V$ denotes the vertex set of $G$ and $E$ denotes the edge set. For $x,y\in V$, we write $x\sim y$ to indicate that $x$ and $y$ are neighbors, i.e. $\left(x,y\right)\in E$. Throughout the paper, we assume that $\#{V}=N$ is finite and $k\geqslant 2$.

The payoff matrix for a generic game with strategies $A$ and $B$ is
\begin{linenomath}
\begin{align}\label{eq:generic2by2}
\bordermatrix{%
 & A & B \cr
A &\ a & \ b \cr
B &\ c & \ d \cr
}\ .
\end{align}
\end{linenomath}
A configuration on $G$, denoted $\xi$, is a function from $V$ to $\left\{0,1\right\}$. If $\xi\left(x\right) =1$, then the player at vertex $x$ is using $A$; otherwise, this player is using $B$. A special case of \eq{generic2by2} is the donation game,
\begin{linenomath}
\begin{align}\label{eq:donationMatrix}
\bordermatrix{%
 & C & D \cr
C &\ b-c & \ -c \cr
D &\ b & \ 0 \cr
}\ .
\end{align}
\end{linenomath}
When we are considering the donation game, $\xi\left(x\right) =1$ indicates a cooperator at vertex $x$ and $\xi\left(x\right) =0$ indicates a defector at vertex $x$. For any such configuration, $\xi$, the conjugate configuration, $\widehat{\xi}$, is defined as $\widehat{\xi}\left(x\right) =1-\xi\left(x\right)$ for $x\in V$. In other words, $\widehat{\xi}\left(x\right) =0$ if $\xi\left(x\right) =1$ and $\widehat{\xi}\left(x\right) =1$ if $\xi\left(x\right) =0$.

For any configuration, $\xi$, on a $k$-regular graph, $G$, and for $x\in V$ and $i,j\in\left\{0,1\right\}$, let
\begin{linenomath}
\begin{subequations}\label{eq:localFrequencyDefinitions}
\begin{align}
f_{i}\left(x,\xi\right) &= \frac{\#\left\{y\in V\ :\ x\sim y\textrm{ and }\xi\left(y\right) =i\right\}}{k} ; \\
f_{ij}\left(x,\xi\right) &= \frac{\#\left\{\left(y,z\right)\in V\times V\ :\ x\sim y\sim z,\ \xi\left(y\right) =i,\textrm{ and }\xi\left(z\right) =j\right\}}{k^{2}} .
\end{align}
\end{subequations}
\end{linenomath}
For any function, $f\left(x,\xi\right)$, let
\begin{linenomath}
\begin{align}\label{si:eq:averageOfFunction}
\overline{f}\left(\xi\right) &:= \frac{1}{N}\sum_{x\in V}f\left(x,\xi\right)
\end{align}
\end{linenomath}
be the arithmetic average of $f$ with respect to the vertices of $G$. (\fig{localFrequencies} in the main text gives an example of how these quantities are calculated.) The arithmetic averages of the functions formed from these local frequencies admit simple probabilistic interpretations: If a random walk is performed on the graph at a starting point chosen uniformly-at-random, then $\overline{f_{1}}\left(\xi\right)$ (resp. $\overline{f_{0}}\left(\xi\right) =1-\overline{f_{1}}\left(\xi\right)$) is the probability that the player at the first step is a cooperator (resp. a defector), and $\overline{f_{10}}\left(\xi\right)$ is the probability that the player at the first step is a cooperator and the player at the second step is a defector. If two independent random walks are performed at the same starting point, then $\overline{f_{1}f_{0}}\left(\xi\right)$ is the probability of finding a cooperator at step one in the first random walk and a defector at step one in the second random walk. If one chooses an enumeration of the vertices and represents $G$ by an adjacency matrix, $\Gamma$, and $\xi$ as a column vector, then
\begin{linenomath}
\begin{subequations}
\begin{align}
\overline{f_{10}}\left(\xi\right) &= \frac{1}{kN}\xi^{\T}\Gamma\widehat{\xi} ; \\
\overline{f_{1}f_{0}}\left(\xi\right) &= \frac{1}{k^{2}N}\xi^{\T}\Gamma^{2}\widehat{\xi} ,
\end{align}
\end{subequations}
\end{linenomath}
which gives a simple, alternative way to calculate each of $\overline{f_{10}}\left(\xi\right)$ and $\overline{f_{1}f_{0}}\left(\xi\right)$.

Let $w\geqslant 0$ be a sufficiently small selection intensity. The effective payoff of an $i$-player at vertex $x$ in configuration $\xi$, denoted $e_{i}^{w}\left(x,\xi\right)$, for the game whose payoffs are given by the generic matrix of \eq{generic2by2}, is defined via
\begin{linenomath}
\begin{subequations} 
\begin{align}
e_{1}^{w}\left(x,\xi\right) &= 1 + wk\left[ a f_{1}\left(x,\xi\right) + b f_{0}\left(x,\xi\right) \right] ; \\
e_{0}^{w}\left(x,\xi\right) &= 1 + wk\left[ c f_{1}\left(x,\xi\right) + d f_{0}\left(x,\xi\right) \right] .
\end{align}
\end{subequations}
\end{linenomath}

The basic measure we use here to define the evolutionary success of a strategy is fixation probability. If $X$ is a strategy (either in $\left\{A,B\right\}$ or in $\left\{C,D\right\}$), let $\mathbf{X}$ denote the monomorphic configuration in which every player uses $X$. For any configuration, $\xi$, and a fixed game, we write $\rho_{\xi ,\mathbf{X}}\left(w\right)$ to denote the probability that strategy $X$ fixates in the population given an initial configuration, $\xi$, and selection intensity, $w$.

In the following sections, we consider DB and BD updating under weak selection ($w\ll 1$).

\subsection{DB updating}

Under DB updating, a player is first selected for death uniformly-at-random from the population. The neighbors of this player then compete to reproduce, with probability proportional to fitness (effective payoff), and the offspring of the reproducing player fills the vacancy. We assume that the strategy of the offspring is inherited from the parent. Therefore, if the player at vertex $x$ dies when the state of the population is $\xi$, then the probability that this vacancy is filled by an $i$-player is
\begin{linenomath}
\begin{align}\label{eq:dbProbability}
\pi_{i}^{w}\left(x,\xi\right) &= \frac{\displaystyle\sum_{y\in V\,:\,y\sim x}e_{i}^{w}\left(y,\xi\right)\mathds{1}_{\xi\left(y\right) =i}}{\displaystyle\sum_{y\in V\,:\,y\sim x}\left[e_{1}^{w}\left(y,\xi\right)\xi\left(y\right) +e_{0}^{w}\left(y,\xi\right)\widehat{\xi}\left(y\right)\right]} .
\end{align}
\end{linenomath}
This DB update rule defines a rate-$N$ pure-jump Markov chain, where $N$ is the size of the population. When $w=0$, this process reduces to the voter model such that, at each update time, a random individual adopts the strategy of a random neighbor \citep[see][]{liggett:S:1985}.

\subsubsection{Critical benefit-to-cost ratios}

Recall that our goal is to determine when, for any configuration, $\xi$, $\rho_{\xi ,\mathbf{C}}\left(w\right) >\rho_{\xi ,\mathbf{C}}\left(0\right)$ for all sufficiently small $w>0$. We first need some technical lemmas:
\begin{lemma}\label{lem:firstLemma}
For any configuration, $\xi$, we have the following first-order expansion as $w\rightarrow 0^{+}$:
\begin{linenomath}
\begin{align}\label{eq:firstLemmaEquation}
\rho_{\xi ,\mathbf{C}}\left(w\right) = \rho_{\xi ,\mathbf{C}}\left(0\right) + w\Bigg( &ak\int_{0}^{\infty}\mathbb{E}_{\xi}^{0}\Big[\overline{f_{0}f_{11}}\left(\xi_{t}\right)\Big]\,dt + bk\int_{0}^{\infty}\mathbb{E}_{\xi}^{0}\Big[\overline{f_{0}f_{10}}\left(\xi_{t}\right)\Big]\,dt \nonumber \\
&-ck\int_{0}^{\infty}\mathbb{E}_{\xi}^{0}\Big[\overline{f_{1}f_{01}}\left(\xi_{t}\right)\Big]\,dt -dk\int_{0}^{\infty}\mathbb{E}_{\xi}^{0}\Big[\overline{f_{1}f_{00}}\left(\xi_{t}\right)\Big]\,dt\Bigg) + O\left(w^{2}\right) .
\end{align}
\end{linenomath}
\end{lemma}
\begin{proof}
By Theorem 3.8 in \citep{chen:AAP:2013}, we have
\begin{linenomath}
\begin{align}\label{eq:firstOrderExpansion}
\rho_{\xi ,\mathbf{C}}\left(w\right) = \rho_{\xi ,\mathbf{C}}\left(0\right) + w\int_{0}^{\infty}\mathbb{E}_{\xi}^{0}\Big[\overline{D}\left(\xi_{t}\right)\Big]\,dt + O\left(w^{2}\right)
\end{align}
\end{linenomath}
whenever $w$ is sufficiently small, where
\begin{linenomath}
\begin{subequations}\label{eq:DandH}
\begin{align}
\overline{D}\left(\xi\right) &= \frac{1}{N}\sum_{x\in V}\Big( \widehat{\xi}\left(x\right) h_{1}\left(x,\xi\right) - \xi\left(x\right) h_{0}\left(x,\xi\right)\Big) ; \\
h_{i}\left(x,\xi\right) &= \frac{d}{dw}\Big\vert_{w=0}\pi_{i}^{w}\left(x,\xi\right) .
\end{align}
\end{subequations}
\end{linenomath}
By the definition of $\pi_{i}^{w}$, \eq{dbProbability}, we have
\begin{linenomath}
\begin{subequations}
\begin{align}
h_{1}\left(x,\xi\right) &= akf_{0}f_{11}\left(x,\xi\right) + bkf_{0}f_{10}\left(x,\xi\right) - ckf_{1}f_{01}\left(x,\xi\right) - dkf_{1}f_{00}\left(x,\xi\right) ; \\
h_{0}\left(x,\xi\right) &= -h_{1}\left(x,\xi\right) ,
\end{align}
\end{subequations}
\end{linenomath}
so \eq{firstLemmaEquation} follows at once from \eq{firstOrderExpansion}, which completes the proof.
\end{proof}

\begin{remark}
The approach of studying fixation probabilities via first-order expansions, as in \eq{firstOrderExpansion}, also appears in \citep{rousset:JTB:2003}, \citep{lessard:JMB:2007}, and \citep{ladret:JTB:2008}. The proof of \eq{firstOrderExpansion} in \citep{chen:AAP:2013}, which is valid under mild assumptions on the game dynamics, was obtained independently and is a particular consequence of a series-like expansion for fixation probabilities. In addition to the identification of the first-order coefficients $\int_{0}^{\infty}\mathbb{E}_{\xi}^{0}\Big[\overline{D}\left(\xi_{t}\right)\Big]\,dt$ in selection strength, $w$, in \eq{firstOrderExpansion}, the proof of this series-like expansion obtains a bound for the $O\left(w^{2}\right)$ error terms that is explicit in selection strength and the rate to reach monomorphic configurations of the underlying game dynamics. Therefore, one can deduce an explicit range of selection strengths such that the comparison of fixation probabilities requires only the sign of $\int_{0}^{\infty}\mathbb{E}_{\xi}^{0}\Big[\overline{D}\left(\xi_{t}\right)\Big]\,dt$. We refer the reader to \citep{wu:PRE:2010} for a further discussion of selection strengths and their consequences for the comparison of fixation probabilities.
\end{remark}

In order to compute the voter-model integrals in \eq{firstLemmaEquation}, we now turn to coalescing random walks on graphs. Suppose that $\left\{B^{x}\right\}_{x\in V}$ is a system of rate-$1$ coalescing random walks on $G$, where, for each $x\in V$, $B^{x}$ starts at $x$. These interacting random walks move independently of one another until they meet, and thereafter they move together. The duality between the voter model and these random walks is given by
\begin{linenomath}
\begin{align}\label{eq:voterDuality}
\mathbb{E}_{\xi}^{0}\left[\prod_{x\in S}\xi_{t}\left(x\right)\right] &= \mathbb{E}\left[\prod_{x\in S}\xi\left(B_{t}^{x}\right)\right]
\end{align}
\end{linenomath}
for each $S\subseteq V$, $t>0$, and strategy configuration, $\xi$. For more information on this duality, including a proof of \eq{voterDuality} and its graphical representation, see \S{III.4} and \S{III.6} in \citep{liggett:S:1985}.

Consider now two discrete-time random walks on $G$, $\left(X_{n}\right)_{n\geqslant 0}$ and $\left(Y_{n}\right)_{n\geqslant 0}$, that start at the same vertex and are independent of $\left\{B^{x}\right\}_{x\in V}$. If the common starting point is $x\in V$, then we write $\mathbb{E}_{x}$ to denote the expectation with respect to this starting point. If the starting point is chosen with respect to the uniform distribution, $\pi$, then this expectation is denoted by $\mathbb{E}_{\pi}$. The random-walk probabilities, $\mathbb{P}_{x}$ and $\mathbb{P}_{\pi}$, are understood in the same way. Since $\displaystyle\sum_{x\in V}\frac{1}{N}\xi\left(x\right)\sum_{y\in V\,:\,y\sim x}\frac{1}{k}\widehat{\xi}\left(y\right) =\mathbb{E}_{\pi}\left[\xi\left(X_{0}\right)\widehat{\xi}\left(X_{1}\right)\right]$, for example, we will use these random walks to save notation when we compute the local frequencies of strategy configurations.

\begin{lemma}\label{lem:secondLemma}
If $f_{\ast 0}:=f_{10}+f_{00}$, then, for any configuration, $\xi$, we have
\begin{linenomath}
\begin{subequations}\label{eq:secondLemmaCalculations}
\begin{align}
\int_{0}^{\infty}\mathbb{E}_{\xi}^{0}\Big[\overline{f_{10}}\left(\xi_{t}\right)\Big]\,dt &= \frac{N\overline{f_{1}}\left(\xi\right)\overline{f_{0}}\left(\xi\right)}{2} ; \label{eq:integralOne} \\
\int_{0}^{\infty}\mathbb{E}_{\xi}^{0}\Big[\overline{f_{1}f_{0}}\left(\xi_{t}\right)\Big]\,dt &= \frac{N\overline{f_{1}}\left(\xi\right)\overline{f_{0}}\left(\xi\right)}{2} - \frac{\overline{f_{10}}\left(\xi\right)}{2} ; \label{eq:integralTwo} \\
\int_{0}^{\infty}\mathbb{E}_{\xi}^{0}\Big[\overline{f_{1}f_{\ast 0}}\left(\xi_{t}\right)\Big]\,dt &= \frac{N\left(1+\frac{1}{k}\right)\overline{f_{1}}\left(\xi\right)\overline{f_{0}}\left(\xi\right)}{2} - \frac{\overline{f_{10}}\left(\xi\right)}{2} - \frac{\overline{f_{1}f_{0}}\left(\xi\right)}{2} . \label{eq:integralThree}
\end{align}
\end{subequations}
\end{linenomath}
\end{lemma}
\begin{proof}
For any configuration, $\xi$, and any $t>0$,
\begin{linenomath}
\begin{align}\label{eq:voterModelEquation}
\mathbb{E}_{\xi}^{0}\Big[\overline{f_{1}}\left(\xi_{t}\right)\overline{f_{0}}\left(\xi_{t}\right)\Big] &= \overline{f_{1}}\left(\xi\right)\overline{f_{0}}\left(\xi\right) - \frac{2}{N}\int_{0}^{t}\mathbb{E}_{\xi}^{0}\Big[\overline{f_{10}}\left(\xi_{s}\right)\Big]\,ds
\end{align}
\end{linenomath}
by Theorem 3.1 in \citep{chen:AIHPPS:2016}. See also Section 3 in that reference for discussions and related results of \eq{voterModelEquation} in terms of coalescing random walks. Moreover, for any vertices $x$ and $y$ with $x\neq y$, we have
\begin{linenomath}
\begin{align}\label{eq:integralEquation}
\mathbb{E}\left[\xi\left(B_{t}^{x}\right)\widehat{\xi}\left(B_{t}^{y}\right)\right] &= e^{-2t}\xi\left(x\right)\widehat{\xi}\left(y\right) \nonumber \\
&\quad +\int_{0}^{t} e^{-2\left(t-s\right)}\left(\sum_{z\in V\,:\,z\sim x}\frac{1}{k}\mathbb{E}\left[\xi\left(B_{s}^{z}\right)\widehat{\xi}\left(B_{s}^{y}\right)\right] + \sum_{z\in V\,:\,z\sim y}\frac{1}{k}\mathbb{E}\left[\xi\left(B_{s}^{x}\right)\widehat{\xi}\left(B_{s}^{z}\right)\right]\right)\,ds ,
\end{align}
\end{linenomath}
which is obtained by considering whether the first epoch time of the bivariate Markov chain $\left(B^{x},B^{y}\right)$ occurs before time $t$ or not. Notice that \eq{integralEquation} is false if $x=y$ since the left-hand side vanishes but the integral term on the right-hand side is, in general, nonzero. This fact needs to be kept in mind when \eq{integralEquation} is applied. Furthermore, using the duality of \eq{voterDuality}, the voter-model integrals in question are
\begin{linenomath}
\begin{subequations}
\begin{align}
\int_{0}^{\infty}\mathbb{E}_{\xi}^{0}\Big[\overline{f_{10}}\left(\xi_{t}\right)\Big]\,dt &= \int_{0}^{\infty}\mathbb{E}_{\pi}\left[\xi\left(B_{t}^{X_{0}}\right)\widehat{\xi}\left(B_{t}^{X_{1}}\right)\right]\,dt ; \label{eq:integralOnePrime} \\
\int_{0}^{\infty}\mathbb{E}_{\xi}^{0}\Big[\overline{f_{1}f_{0}}\left(\xi_{t}\right)\Big]\,dt &= \int_{0}^{\infty}\mathbb{E}_{\pi}\left[\xi\left(B_{t}^{X_{0}}\right)\widehat{\xi}\left(B_{t}^{X_{2}}\right)\right]\,dt ; \label{eq:integralTwoPrime} \\
\int_{0}^{\infty}\mathbb{E}_{\xi}^{0}\Big[\overline{f_{1}f_{\ast 0}}\left(\xi_{t}\right)\Big]\,dt &=\int_{0}^{\infty}\mathbb{E}_{\pi}\left[\xi\left(B_{t}^{X_{0}}\right)\widehat{\xi}\left(B_{t}^{X_{3}}\right)\right]\,dt . \label{eq:integralThreePrime}
\end{align}
\end{subequations}
\end{linenomath}
We are now in a position to establish \eq{secondLemmaCalculations}. By letting $t\rightarrow\infty$ in \eq{voterModelEquation}, we obtain \eq{integralOne} since $\overline{f_{1}}\cdot\overline{f_{0}}$ vanishes at monomorphic configurations. Since the graph has no self-loops, we have $X_{0}\neq X_{1}$ almost surely, thus, by \eq{integralEquation} and the reversibility of the chain $\left(X_{n}\right)_{n\geqslant 0}$ under $\mathbb{P}_{\pi}$, we have
\begin{linenomath}
\begin{align}
\mathbb{E}_{\pi}\left[\xi\left(B_{t}^{X_{0}}\right)\widehat{\xi}\left(B_{t}^{X_{1}}\right)\right] &= e^{-2t}\mathbb{E}_{\pi}\left[\xi\left(X_{0}\right)\widehat{\xi}\left(X_{1}\right)\right] \nonumber \\
&\quad +\int_{0}^{t} e^{-2\left(t-s\right)} \left( \mathbb{E}_{\pi}\left[\xi\left(B_{s}^{Y_{1}}\right)\widehat{\xi}\left(B_{s}^{X_{1}}\right)\right] + \mathbb{E}_{\pi}\left[\xi\left(B_{s}^{X_{0}}\right)\widehat{\xi}\left(B_{s}^{X_{2}}\right)\right] \right) \,ds \nonumber \\
&= e^{-2t}\mathbb{E}_{\pi}\left[\xi\left(X_{0}\right)\widehat{\xi}\left(X_{1}\right)\right] + \int_{0}^{t} 2e^{-2\left(t-s\right)}\mathbb{E}_{\pi}\left[\xi\left(B_{s}^{X_{0}}\right)\widehat{\xi}\left(B_{s}^{X_{2}}\right)\right]\,ds .
\end{align}
\end{linenomath}
Integrating both sides of this equation with respect to $t$ over $\left(0,\infty\right)$ implies that
\begin{linenomath}
\begin{align}
\int_{0}^{\infty}\mathbb{E}_{\pi}\left[\xi\left(B_{t}^{X_{0}}\right)\widehat{\xi}\left(B_{t}^{X_{2}}\right)\right]\,dt &= \int_{0}^{\infty}\mathbb{E}_{\pi}\left[\xi\left(B_{t}^{X_{0}}\right)\widehat{\xi}\left(B_{t}^{X_{1}}\right)\right]\,dt - \frac{\mathbb{E}_{\pi}\left[\xi\left(X_{0}\right)\widehat{\xi}\left(X_{1}\right)\right]}{2} ,
\end{align}
\end{linenomath}
which, by Eqs. (\ref{eq:integralOne}), (\ref{eq:integralOnePrime}), and (\ref{eq:integralTwoPrime}), gives \eq{integralTwo}.

The proof of the one remaining equation, \eq{integralThree}, is similar except that we have to take into account the fact that $\mathbb{P}_{\pi}\left(X_{0}=X_{2}\right) >0$ when applying \eq{integralEquation}. By reversibility, $\left(Y_{1},X_{0},X_{1},X_{2}\right)$ and $\left(X_{3},X_{2},X_{1},X_{0}\right)$ have the same distribution under $\mathbb{P}_{\pi}$. Therefore, by \eq{integralEquation}, it follows that
\begin{linenomath}
\begin{align}
\mathbb{E}_{\pi} &\left[\xi\left(B_{t}^{X_{0}}\right)\widehat{\xi}\left(B_{t}^{X_{2}}\right)\right] \nonumber \\
&= e^{-2t}\mathbb{E}_{\pi}\left[\xi\left(X_{0}\right)\widehat{\xi}\left(X_{2}\right)\right] \nonumber \\
&\quad +\int_{0}^{t}e^{-2\left(t-s\right)}\left(\mathbb{E}_{\pi}\left[\xi\left(B_{s}^{Y_{1}}\right)\widehat{\xi}\left(B_{s}^{X_{2}}\right)\mathds{1}_{\left\{X_{0}\neq X_{2}\right\}}\right] + \mathbb{E}_{\pi}\left[\xi\left(B_{s}^{X_{0}}\right)\widehat{\xi}\left(B_{s}^{X_{3}}\right)\mathds{1}_{\left\{X_{0}\neq X_{2}\right\}}\right]\right)\,ds \nonumber \\
&= e^{-2t}\mathbb{E}_{\pi}\left[\xi\left(X_{0}\right)\widehat{\xi}\left(X_{2}\right)\right] \nonumber \\
&\quad +\int_{0}^{t}e^{-2\left(t-s\right)}\left(\mathbb{E}_{\pi}\left[\xi\left(B_{s}^{X_{3}}\right)\widehat{\xi}\left(B_{s}^{X_{0}}\right)\mathds{1}_{\left\{X_{0}\neq X_{2}\right\}}\right] + \mathbb{E}_{\pi}\left[\xi\left(B_{s}^{X_{0}}\right)\widehat{\xi}\left(B_{s}^{X_{3}}\right)\mathds{1}_{\left\{X_{0}\neq X_{2}\right\}}\right]\right)\,ds .
\end{align}
\end{linenomath}
Integrating both sides of this equation with respect to $t$ over $\left(0,\infty\right)$ yields
\begin{linenomath}
\begin{align}\label{eq:0toInfty}
\int_{0}^{\infty}\mathbb{E}_{\pi}\left[\xi\left(B_{t}^{X_{0}}\right)\widehat{\xi}\left(B_{t}^{X_{2}}\right)\right]\,dt &= \frac{\mathbb{E}_{\pi}\left[\xi\left(X_{0}\right)\widehat{\xi}\left(X_{2}\right)\right]}{2} + \frac{1}{2}\int_{0}^{\infty}\mathbb{E}_{\pi}\left[\xi\left(B_{t}^{X_{3}}\right)\widehat{\xi}\left(B_{t}^{X_{0}}\right)\mathds{1}_{\left\{X_{0}\neq X_{2}\right\}}\right]\,dt \nonumber \\
&\quad + \frac{1}{2}\int_{0}^{\infty}\mathbb{E}_{\pi}\left[\xi\left(B_{t}^{X_{0}}\right)\widehat{\xi}\left(B_{t}^{X_{3}}\right)\mathds{1}_{\left\{X_{0}\neq X_{2}\right\}}\right]\,dt ,
\end{align}
\end{linenomath}
from which we obtain
\begin{linenomath}
\begin{align}\label{eq:finalEquation}
\int_{0}^{\infty} &\mathbb{E}_{\pi}\left[\xi\left(B_{t}^{X_{0}}\right)\widehat{\xi}\left(B_{t}^{X_{3}}\right)\right]\,dt \nonumber \\
&= \frac{1}{2}\int_{0}^{\infty}\mathbb{E}_{\pi}\left[\xi\left(B_{t}^{X_{3}}\right)\widehat{\xi}\left(B_{t}^{X_{0}}\right)\right]\,dt + \frac{1}{2}\int_{0}^{\infty}\mathbb{E}_{\pi}\left[\xi\left(B_{t}^{X_{0}}\right)\widehat{\xi}\left(B_{t}^{X_{3}}\right)\right]\,dt \nonumber \\
&= \frac{1}{2}\int_{0}^{\infty}\mathbb{E}_{\pi}\left[\xi\left(B_{t}^{X_{3}}\right)\widehat{\xi}\left(B_{t}^{X_{0}}\right)\mathds{1}_{\left\{X_{0}\neq X_{2}\right\}}\right]\,dt + \frac{1}{2}\int_{0}^{\infty}\mathbb{E}_{\pi}\left[\xi\left(B_{t}^{X_{3}}\right)\widehat{\xi}\left(B_{t}^{X_{0}}\right)\mathds{1}_{\left\{X_{0}=X_{2}\right\}}\right]\,dt \nonumber \\
&\quad +\frac{1}{2}\int_{0}^{\infty}\mathbb{E}_{\pi}\left[\xi\left(B_{t}^{X_{0}}\right)\widehat{\xi}\left(B_{t}^{X_{3}}\right)\mathds{1}_{\left\{X_{0}\neq X_{2}\right\}}\right]\,dt + \frac{1}{2}\int_{0}^{\infty}\mathbb{E}_{\pi}\left[\xi\left(B_{t}^{X_{0}}\right)\widehat{\xi}\left(B_{t}^{X_{3}}\right)\mathds{1}_{\left\{X_{0}=X_{2}\right\}}\right]\,dt \nonumber \\
&= \int_{0}^{\infty}\mathbb{E}_{\pi}\left[\xi\left(B_{t}^{X_{0}}\right)\widehat{\xi}\left(B_{t}^{X_{2}}\right)\right]\,dt - \frac{\mathbb{E}_{\pi}\left[\xi\left(X_{0}\right)\widehat{\xi}\left(X_{2}\right)\right]}{2} \nonumber \\
&\quad +\frac{1}{2}\int_{0}^{\infty}\mathbb{E}_{\pi}\left[\xi\left(B_{t}^{X_{3}}\right)\widehat{\xi}\left(B_{t}^{X_{2}}\right)\mathds{1}_{\left\{X_{0}=X_{2}\right\}}\right] \,dt + \frac{1}{2}\int_{0}^{\infty} \mathbb{E}_{\pi}\left[\xi\left(B_{t}^{X_{2}}\right)\widehat{\xi}\left(B_{t}^{X_{3}}\right)\mathds{1}_{\left\{X_{0}=X_{2}\right\}}\right] \,dt \nonumber \\
&= \int_{0}^{\infty}\mathbb{E}_{\pi}\left[\xi\left(B_{t}^{X_{0}}\right)\widehat{\xi}\left(B_{t}^{X_{2}}\right)\right]\,dt - \frac{\mathbb{E}_{\pi}\left[\xi\left(X_{0}\right)\widehat{\xi}\left(X_{2}\right)\right]}{2} \nonumber \\
&\quad +\frac{1}{2k}\int_{0}^{\infty}\mathbb{E}_{\pi}\left[\xi\left(B_{t}^{X_{1}}\right)\widehat{\xi}\left(B_{t}^{X_{0}}\right)\right] \,dt + \frac{1}{2k}\int_{0}^{\infty} \mathbb{E}_{\pi}\left[\xi\left(B_{t}^{X_{0}}\right)\widehat{\xi}\left(B_{t}^{X_{1}}\right)\right] \,dt \nonumber \\
&=  \int_{0}^{\infty}\mathbb{E}_{\pi}\left[\xi\left(B_{t}^{X_{0}}\right)\widehat{\xi}\left(B_{t}^{X_{2}}\right)\right]\,dt - \frac{\mathbb{E}_{\pi}\left[\xi\left(X_{0}\right)\widehat{\xi}\left(X_{2}\right)\right]}{2} + \frac{1}{k}\int_{0}^{\infty} \mathbb{E}_{\pi}\left[\xi\left(B_{t}^{X_{0}}\right)\widehat{\xi}\left(B_{t}^{X_{1}}\right)\right] \,dt .
\end{align}
\end{linenomath}
The first and last equalities follow from reversibility, the third equality from \eq{0toInfty}, and the fourth equality from the Markov property of $\left(X_{n}\right)_{n\geqslant 0}$ at $n=0,2$ and the fact that $\mathbb{P}_{x}\left(X_{0}=X_{2}\right) =1/k$ since the graph is regular. \eq{integralThree} then follows from Eqs. (\ref{eq:integralOne}), (\ref{eq:integralTwo}), (\ref{eq:integralThreePrime}), and (\ref{eq:finalEquation}).
\end{proof}

We are now in a position to prove the first of our main results:
\begin{theorem}\label{thm:dbTheorem}
In the donation game, for any configuration, $\xi$, we have the following expansion as $w\rightarrow 0^{+}$:
\begin{linenomath}
\begin{align}\label{eq:dbTheoremEquation}
\rho_{\xi ,\mathbf{C}}\left(w\right) = \rho_{\xi ,\mathbf{C}}\left(0\right) + \frac{w}{2}\Big\{ &b\left[N\overline{f_{1}}\left(\xi\right)\overline{f_{0}}\left(\xi\right) -k\overline{f_{10}}\left(\xi\right) - k\overline{f_{1}f_{0}}\left(\xi\right)\right] \nonumber \\
&- c\left[kN\overline{f_{1}}\left(\xi\right)\overline{f_{0}}\left(\xi\right) -k\overline{f_{10}}\left(\xi\right)\right] \Big\} + O\left(w^{2}\right) .
\end{align}
\end{linenomath}
\end{theorem}
\begin{proof}
By \lem{firstLemma}, it suffices to obtain the coefficient of $w$, i.e. the first order term, on the right-hand side of \eq{firstLemmaEquation}. Since the game under consideration is the donation game, a simple calculation gives
\begin{linenomath}
\begin{align}\label{eq:dbProofEquation}
\int_{0}^{\infty}\mathbb{E}_{\xi}^{0}\Big[\overline{D}\left(\xi_{t}\right)\Big]\,dt &= bk\left(\int_{0}^{\infty}\mathbb{E}_{\xi}^{0}\Big[\overline{f_{1}f_{\ast 0}}\left(\xi_{t}\right)\Big]\,dt - \int_{0}^{\infty}\mathbb{E}_{\xi}^{0}\Big[\overline{f_{10}}\left(\xi_{t}\right)\Big]\,dt\right) - ck\int_{0}^{\infty}\mathbb{E}_{\xi}^{0}\Big[\overline{f_{1}f_{0}}\left(\xi_{t}\right)\Big]\,dt .
\end{align}
\end{linenomath}
Therefore, \eq{dbTheoremEquation} follows from the calculations of \lem{secondLemma}, which completes the proof.
\end{proof}

From \thm{dbTheorem}, we see that, for small $w>0$,
\begin{linenomath}
\begin{align}\label{eq:dbTheoremConclusion}
\rho_{\xi ,\mathbf{C}}\left(w\right) > \rho_{\xi ,\mathbf{C}}\left(0\right) &\iff \frac{b}{c} > \frac{k\left(N\overline{f_{1}}\left(\xi\right)\overline{f_{0}}\left(\xi\right) -\overline{f_{10}}\left(\xi\right)\right)}{N\overline{f_{1}}\left(\xi\right)\overline{f_{0}}\left(\xi\right) -k\overline{f_{10}}\left(\xi\right) - k\overline{f_{1}f_{0}}\left(\xi\right)} =: \left(\frac{b}{c}\right)_{\xi}^{\ast} ,
\end{align}
\end{linenomath}
which gives the critical benefit-to-cost ratio of \eq{ratioGeneral}.

\subsubsection{Structure coefficients}

We now turn to a generalization of the critical benefit-to-cost ratio for arbitrary $2\times 2$ games in which the payoff matrix is given by \eq{generic2by2}. Our main result is the following:

\begin{theorem}\label{thm:dbSigmaTheorem}
$\rho_{\xi ,\mathbf{A}}\left(w\right) > \rho_{\widehat{\xi},\mathbf{B}}\left(w\right)$ for all sufficiently small $w>0$ if and only if
\begin{linenomath}
\begin{align}\label{eq:sigmaTheoremEquation}
\left(a-d\right) &\left[N\overline{f_{1}}\left(\xi\right)\overline{f_{0}}\left(\xi\right)\left(1+\frac{1}{k}\right) -2\overline{f_{10}}\left(\xi\right) -\overline{f_{1}f_{0}}\left(\xi\right)\right] \nonumber \\
&+\left(b-c\right)\left[N\overline{f_{1}}\left(\xi\right)\overline{f_{0}}\left(\xi\right)\left(1-\frac{1}{k}\right) + \overline{f_{1}f_{0}}\left(\xi\right)\right] > 0 .
\end{align}
\end{linenomath}
\end{theorem}
\begin{proof}
By the neutrality of the voter model, we have $\rho_{\xi ,\mathbf{A}}\left(0\right) +\rho_{\widehat{\xi},\mathbf{A}}\left(0\right) =1$, thus
\begin{linenomath}
\begin{align}
\rho_{\xi ,\mathbf{A}}\left(w\right) > \rho_{\widehat{\xi},\mathbf{B}}\left(w\right) &\iff \rho_{\xi ,\mathbf{A}}\left(w\right) > 1 - \rho_{\widehat{\xi},\mathbf{A}}\left(w\right) \nonumber \\
&\iff \Big( \rho_{\xi ,\mathbf{A}}\left(w\right) - \rho_{\xi ,\mathbf{A}}\left(0\right) \Big) + \Big( \rho_{\widehat{\xi} ,\mathbf{A}}\left(w\right) - \rho_{\widehat{\xi} ,\mathbf{A}}\left(0\right) \Big) > 0
\end{align}
\end{linenomath}
By the first-order expansion of \eq{firstOrderExpansion}, it follows that, for small $w>0$,
\begin{linenomath}
\begin{align}
\rho_{\xi ,\mathbf{A}}\left(w\right) > \rho_{\widehat{\xi},\mathbf{B}}\left(w\right) &\iff \int_{0}^{\infty}\mathbb{E}_{\xi}^{0}\Big[\overline{D}\left(\xi_{t}\right)\Big]\,dt + \int_{0}^{\infty}\mathbb{E}_{\widehat{\xi}}^{0}\Big[\overline{D}\left(\xi_{t}\right)\Big]\,dt > 0 .
\end{align}
\end{linenomath}
By Eqs. (\ref{eq:firstLemmaEquation}) and (\ref{eq:firstOrderExpansion}) and the neutrality of the voter model, we have
\begin{linenomath}
\begin{align}\label{eq:sigmaExpansion}
\int_{0}^{\infty} &\mathbb{E}_{\xi}^{0}\Big[\overline{D}\left(\xi_{t}\right)\Big]\,dt + \int_{0}^{\infty}\mathbb{E}_{\widehat{\xi}}^{0}\Big[\overline{D}\left(\xi_{t}\right)\Big]\,dt \nonumber \\
&= \left(a-d\right) k\left(\int_{0}^{\infty}\mathbb{E}_{\xi}^{0}\Big[\overline{f_{0}f_{11}}\left(\xi_{t}\right)\Big]\,dt +\int_{0}^{\infty}\mathbb{E}_{\xi}^{0}\Big[\overline{f_{1}f_{00}}\left(\xi_{t}\right)\Big]\,dt\right) \nonumber \\
&\quad + \left(b-c\right) k\left(\int_{0}^{\infty}\mathbb{E}_{\xi}^{0}\Big[\overline{f_{0}f_{10}}\left(\xi_{t}\right)\Big]\,dt +\int_{0}^{\infty}\mathbb{E}_{\xi}^{0}\Big[\overline{f_{1}f_{01}}\left(\xi_{t}\right)\Big]\,dt\right) ,
\end{align}
\end{linenomath}
and all that remains is to determine the coefficients of $a-d$ and $b-c$ in \eq{sigmaExpansion}. By considering \eq{sigmaExpansion} with the payoff values of the donation game rather than an arbitrary $2\times 2$ game, we obtain
\begin{linenomath}
\begin{align}
\int_{0}^{\infty} &\mathbb{E}_{\xi}^{0}\Big[\overline{D}\left(\xi_{t}\right)\Big]\,dt + \int_{0}^{\infty}\mathbb{E}_{\widehat{\xi}}^{0}\Big[\overline{D}\left(\xi_{t}\right)\Big]\,dt \nonumber \\
&= \left(b-c\right) k\left(\int_{0}^{\infty}\mathbb{E}_{\xi}^{0}\Big[\overline{f_{0}f_{11}}\left(\xi_{t}\right)\Big]\,dt +\int_{0}^{\infty}\mathbb{E}_{\xi}^{0}\Big[\overline{f_{1}f_{00}}\left(\xi_{t}\right)\Big]\,dt\right) \nonumber \\
&\quad - \left(b+c\right) k\left(\int_{0}^{\infty}\mathbb{E}_{\xi}^{0}\Big[\overline{f_{0}f_{10}}\left(\xi_{t}\right)\Big]\,dt +\int_{0}^{\infty}\mathbb{E}_{\xi}^{0}\Big[\overline{f_{1}f_{01}}\left(\xi_{t}\right)\Big]\,dt\right) .
\end{align}
\end{linenomath}
Thus, using \thm{dbTheorem}, we see that, when $b=-c$,
\begin{linenomath}
\begin{align}
\int_{0}^{\infty}\mathbb{E}_{\xi}^{0}\Big[\overline{f_{0}f_{11}}\left(\xi_{t}\right)\Big]\,dt +\int_{0}^{\infty}\mathbb{E}_{\xi}^{0}\Big[\overline{f_{1}f_{00}}\left(\xi_{t}\right)\Big]\,dt &= \frac{1}{2}\left[N\overline{f_{1}}\left(\xi\right)\overline{f_{0}}\left(\xi\right)\left(1+\frac{1}{k}\right) -2\overline{f_{10}}\left(\xi\right) -\overline{f_{1}f_{0}}\left(\xi\right)\right] ,
\end{align}
\end{linenomath}
and, when $b=c$,
\begin{linenomath}
\begin{align}
\int_{0}^{\infty}\mathbb{E}_{\xi}^{0}\Big[\overline{f_{0}f_{10}}\left(\xi_{t}\right)\Big]\,dt +\int_{0}^{\infty}\mathbb{E}_{\xi}^{0}\Big[\overline{f_{1}f_{01}}\left(\xi_{t}\right)\Big]\,dt &= \frac{1}{2}\left[N\overline{f_{1}}\left(\xi\right)\overline{f_{0}}\left(\xi\right)\left(1-\frac{1}{k}\right) + \overline{f_{1}f_{0}}\left(\xi\right)\right] ,
\end{align}
\end{linenomath}
from which we obtain \eq{sigmaTheoremEquation}.
\end{proof}

In other words, $\rho_{\xi ,\mathbf{A}}\left(w\right) > \rho_{\widehat{\xi},\mathbf{B}}\left(w\right)$ for all sufficiently small $w>0$ if and only if
\begin{linenomath}
\begin{align}
\sigma_{\xi}a + b &> c + \sigma_{\xi}d ,
\end{align}
\end{linenomath}
where $\sigma_{\xi}$ is the structure coefficient given by
\begin{linenomath}
\begin{align}
\sigma_{\xi} &= \frac{N\left(1+\frac{1}{k}\right)\overline{f_{1}}\cdot\overline{f_{0}} -2\overline{f_{10}}-\overline{f_{1}f_{0}}}{N\left(1-\frac{1}{k}\right)\overline{f_{1}}\cdot\overline{f_{0}}+\overline{f_{1}f_{0}}} .
\end{align}
\end{linenomath}
A simple calculation shows that
\begin{linenomath}
\begin{align}
\sigma_{\xi} &= \frac{\left(\frac{b}{c}\right)_{\xi}^{\ast}+1}{\left(\frac{b}{c}\right)_{\xi}^{\ast}-1} ,
\end{align}
\end{linenomath}
and, moreover, when the payoffs for the game are given by \eq{donationMatrix}, this result reduces to \eq{dbTheoremConclusion}.

\subsection{BD updating}

Under BD updating, a player is first chosen to reproduce with probability proportional to fitness (effective payoff). A neighbor of the reproducing player is then chosen uniformly-at-random for death, and the offspring of the reproducing player fills this vacancy. The rate at which the player at vertex $x$ is replaced by an $i$-player is then
\begin{linenomath}
\begin{align}
\pi_{i}^{w}\left(x,\xi\right) &= \frac{\displaystyle\sum_{y\in V\,:\,y\sim x}e_{i}^{w}\left(y,\xi\right)\mathds{1}_{\left\{\xi\left(y\right) =i\right\}}}{\displaystyle k\sum_{z\in V}\left[e_{1}^{w}\left(z,\xi\right)\xi\left(z\right) +e_{0}^{w}\left(z,\xi\right)\widehat{\xi}\left(z\right)\right]} .
\end{align}
\end{linenomath}
The neutral version of this process ($w=0$) is a perturbation of the voter model, and we can use techniques similar to those used for DB updating to establish our main results for BD updating.

\subsubsection{Critical benefit-to-cost ratios}

Again, we first need a technical lemma:

\begin{lemma}
For $i,j,l\in\left\{0,1\right\}$, $x\in V$, and $\xi$ a configuration of cooperators and defectors, let
\begin{linenomath}
\begin{align}
f_{ijl}\left(x,\xi\right) &= \frac{\#\left\{\left(y,z,v\right)\in V\times V\times V\ :\ x\sim y\sim z\sim v,\ \xi\left(y\right) =i,\ \xi\left(z\right) =j,\textrm{ and } \xi\left(v\right) =l\right\}}{k^{3}} ,
\end{align}
\end{linenomath}
and form the averages $\overline{f_{ijl}}\left(\xi\right)$ via Eq. (\ref{si:eq:averageOfFunction}). For any configuration, $\xi$, we have the following first-order expansion as $w\rightarrow 0^{+}$:
\begin{linenomath}
\begin{align}\label{eq:bdLemmaEquation}
\rho_{\xi ,\mathbf{C}}\left(w\right) = \rho_{\xi ,\mathbf{C}}\left(0\right) + w\Bigg( &ak\int_{0}^{\infty}\mathbb{E}_{\xi}^{0}\Big[\overline{f_{110}}\left(\xi_{t}\right)\Big]\,dt + bk\int_{0}^{\infty}\mathbb{E}_{\xi}^{0}\Big[\overline{f_{010}}\left(\xi_{t}\right)\Big]\,dt \nonumber \\
&-ck\int_{0}^{\infty}\mathbb{E}_{\xi}^{0}\Big[\overline{f_{101}}\left(\xi_{t}\right)\Big]\,dt -dk\int_{0}^{\infty}\mathbb{E}_{\xi}^{0}\Big[\overline{f_{100}}\left(\xi_{t}\right)\Big]\,dt\Bigg) + O\left(w^{2}\right) .
\end{align}
\end{linenomath}
\end{lemma}
\begin{proof}
The first-order expansion of \eq{firstOrderExpansion} is valid under BD updating as well \citep[see][Theorem 3.8]{chen:AAP:2013}, except that the function $\overline{D}\left(\xi\right)$ of \eq{DandH} is defined in terms of the rates $\pi_{i}^{w}$ for BD updating rather than for DB updating. Writing $e^{w}\left(y,\xi\right) =e_{i}^{w}\left(y,\xi\right)$ whenever $\xi\left(y\right) =i$, we find that
\begin{linenomath}
\begin{align}\label{eq:bdfunctionD}
\overline{D}\left(\xi\right) &= \frac{d}{dw}\Big\vert_{w=0}\left(\frac{1}{N}\sum_{x\in V}\pi_{1}^{w}\left(x,\xi\right)\widehat{\xi}\left(x\right) - \frac{1}{N}\sum_{x\in V}\pi_{0}^{w}\left(x,\xi\right)\xi\left(x\right)\right) \nonumber \\
&= \sum_{x\in V}\xi\left(x\right)\left(\frac{1}{N}\frac{de^{w}\left(x,\xi\right)}{dw}\Big\vert_{w=0}-\frac{1}{N^{2}}\sum_{z\in V}\frac{de^{w}\left(z,\xi\right)}{dw}\Big\vert_{w=0}\right) \nonumber \\
&\quad -\sum_{x\in V}\xi\left(x\right)\left(\frac{1}{Nk}\sum_{y\in V\,:\,y\sim x}\frac{de^{w}\left(y,\xi\right)}{dw}\Big\vert_{w=0}-\frac{1}{N^{2}}\sum_{z\in V}\frac{de^{w}\left(z,\xi\right)}{dw}\Big\vert_{w=0}\right) \nonumber \\
&= k\left(a\overline{f_{11}}\left(\xi\right) +b\overline{f_{10}}\left(\xi\right)\right) - k\left(a\overline{f_{111}}\left(\xi\right) +b\overline{f_{110}}\left(\xi\right) +c\overline{f_{101}}\left(\xi\right) +d\overline{f_{100}}\left(\xi\right)\right) \nonumber \\
&= k\left(a\overline{f_{110}}\left(\xi\right) +b\overline{f_{010}}\left(\xi\right) -c\overline{f_{101}}\left(\xi\right) -d\overline{f_{100}}\left(\xi\right)\right) .
\end{align}
\end{linenomath}
We then obtain \eq{bdLemmaEquation} by applying this calculation to the first-order expansion of \eq{firstOrderExpansion}.
\end{proof}

Our main result for BD updating is the following:
\begin{theorem}
In the donation game, for any configuration, $\xi$, we have the following expansion as $w\rightarrow 0^{+}$:
\begin{linenomath}
\begin{align}\label{eq:bdTheoremEquation}
\rho_{\xi ,\mathbf{C}}\left(w\right) = \rho_{\xi ,\mathbf{C}}\left(0\right) - \frac{wk}{2}\Big\{ b\overline{f_{10}}\left(\xi\right) + cN\overline{f_{1}}\left(\xi\right)\overline{f_{0}}\left(\xi\right) \Big\} + O\left(w^{2}\right) .
\end{align}
\end{linenomath}
\end{theorem}
\begin{proof}
For the donation game, the function $\overline{D}$ of \eq{bdfunctionD} simplifies to
\begin{linenomath}
\begin{align}
\overline{D}\left(\xi\right) &= -kb\overline{f_{10}}\left(\xi\right) + kb\overline{f_{1}f_{0}}\left(\xi\right) - kc\overline{f_{10}}\left(\xi\right) .
\end{align}
\end{linenomath}
Therefore, by the calculations of \lem{secondLemma}, we see that
\begin{linenomath}
\begin{align}
\int_{0}^{\infty}\mathbb{E}_{\xi}^{0}\Big[\overline{D}\left(\xi_{t}\right)\Big]\,dt &= -\frac{k}{2}\left[ b\overline{f_{10}}\left(\xi\right) + cN\overline{f_{1}}\left(\xi\right)\overline{f_{0}}\left(\xi\right) \right] ,
\end{align}
\end{linenomath}
which gives \eq{bdTheoremEquation} and completes the proof.
\end{proof}

Since $b\overline{f_{10}}\left(\xi\right) + cN\overline{f_{1}}\left(\xi\right)\overline{f_{0}}\left(\xi\right) >0$ for each mixed state, $\xi$, it follows that $\rho_{\xi,\mathbf{C}}\left(w\right) <\rho_{\xi ,\mathbf{C}}\left(0\right)$ for all sufficiently small $w>0$ whenever $\xi$ is not an absorbing state, so cooperation is always suppressed by weak selection under BD updating.

\subsubsection{Structure coefficients}

Although cooperation is never favored by weak selection under BD updating, we can still write down a condition for selection to favor strategy $A$ in an arbitrary $2\times 2$ game whose payoff matrix is given by \eq{generic2by2}:
\begin{theorem}
$\rho_{\xi ,\mathbf{A}}\left(w\right) > \rho_{\widehat{\xi},\mathbf{B}}\left(w\right)$ for all sufficiently small $w>0$ if and only if
\begin{linenomath}
\begin{align}\label{eq:bdSigmaEquation}
\left(a-d\right)\left[ N\overline{f_{1}}\left(\xi\right)\overline{f_{0}}\left(\xi\right) - \overline{f_{10}}\left(\xi\right)\right] + \left(b-c\right)\left[N\overline{f_{1}}\left(\xi\right)\overline{f_{0}}\left(\xi\right) + \overline{f_{10}}\left(\xi\right)\right] > 0 .
\end{align}
\end{linenomath}
\end{theorem}
\begin{proof}
The same argument given in the proof of \thm{dbSigmaTheorem} shows that \eq{bdSigmaEquation} is equivalent to
\begin{linenomath}
\begin{align}
\int_{0}^{\infty} &\mathbb{E}_{\xi}^{0}\Big[\overline{D}\left(\xi_{t}\right)\Big]\,dt + \int_{0}^{\infty}\mathbb{E}_{\widehat{\xi}}^{0}\Big[\overline{D}\left(\xi_{t}\right)\Big]\,dt \nonumber \\
&= \left(a-d\right) k\left(\int_{0}^{\infty}\mathbb{E}_{\xi}^{0}\Big[\overline{f_{110}}\left(\xi_{t}\right)\Big]\,dt +\int_{0}^{\infty}\mathbb{E}_{\xi}^{0}\Big[\overline{f_{100}}\left(\xi_{t}\right)\Big]\,dt\right) \nonumber \\
&\quad + \left(b-c\right) k\left(\int_{0}^{\infty}\mathbb{E}_{\xi}^{0}\Big[\overline{f_{010}}\left(\xi_{t}\right)\Big]\,dt +\int_{0}^{\infty}\mathbb{E}_{\xi}^{0}\Big[\overline{f_{101}}\left(\xi_{t}\right)\Big]\,dt\right) > 0 .
\end{align}
\end{linenomath}
Solving for the coefficients of $a-d$ and $b-c$ as in the proof of \thm{dbSigmaTheorem} gives \eq{bdSigmaEquation}.
\end{proof}

Written differently, $\rho_{\xi ,\mathbf{A}}\left(w\right) > \rho_{\widehat{\xi},\mathbf{B}}\left(w\right)$ for all sufficiently small $w>0$ if and only if
\begin{linenomath}
\begin{align}
\sigma_{\xi}a + b &> c + \sigma_{\xi}d ,
\end{align}
\end{linenomath}
where $\sigma_{\xi}$ is the structure coefficient given by
\begin{linenomath}
\begin{align}
\sigma_{\xi} &= \frac{N\overline{f_{1}}\left(\xi\right)\overline{f_{0}}\left(\xi\right) -\overline{f_{10}}\left(\xi\right)}{N\overline{f_{1}}\left(\xi\right)\overline{f_{0}}\left(\xi\right) +\overline{f_{10}}\left(\xi\right)} .
\end{align}
\end{linenomath}

\subsection{Strategic placement of cooperators for DB updating}

We turn now to the consequences of \thm{dbTheorem} for DB updating.

\begin{proposition}\label{prop:uniformConvergence}
Let $k\geqslant 2$ be fixed. In the limit of large population size, $N\rightarrow\infty$, the critical benefit-to-cost ratio converges uniformly to $k$ over all $k$-regular graphs, $G$, of size $N$ and all configurations, $\xi$, on $G$.
\end{proposition}
Proposition \ref{prop:uniformConvergence} follows immediately from the following technical result:
\begin{lemma}
For fixed $k\geqslant 2$ and for $N>4k^{2}+1$ such that there exists a $k$-regular graph with $N$ vertices,
\begin{linenomath}
\begin{align}
\max_{G}\max_{\xi}\left| \left(\frac{b}{c}\right)_{\xi}^{\ast} - k \right| \leqslant \frac{k\left(2k+1\right)}{\left(N-1\right)^{1/2}-2k} ,
\end{align}
\end{linenomath}
where $G$ ranges over all $k$-regular graphs on $N$ vertices, and, for each $G$, $\xi$ ranges over all mixed configurations.
\end{lemma}
\begin{proof}
By the Cauchy-Schwarz inequality and the reversibility of the random walk, both $\overline{f_{10}}\left(\xi\right)$ and $\overline{f_{1}f_{0}}\left(\xi\right)$ are bounded by $\left(\overline{f_{1}}\left(\xi\right)\overline{f_{0}}\left(\xi\right)\right)^{1/2}$. Therefore, for any such $G$ and any mixed $\xi$, it follows from \eq{ratioGeneral} that
\begin{linenomath}
\begin{align}
\max_{G}\max_{\xi}\left| \left(\frac{b}{c}\right)_{\xi}^{\ast} - k \right| \leqslant \max_{0<n<N} \frac{k\left(2k+1\right)\left(\frac{n\left(N-n\right)}{N^{2}}\right)^{1/2}}{N\left(\frac{n\left(N-n\right)}{N^{2}}\right) -2k\left(\frac{n\left(N-n\right)}{N^{2}}\right)^{1/2}} = \frac{k\left(2k+1\right)}{\left(N-1\right)^{1/2}-2k} ,
\end{align}
\end{linenomath}
which completes the proof.
\end{proof}

\begin{proposition}
For all configurations obtained by placing an arbitrary (but fixed) number of cooperators uniformly at random, the critical benefit-to cost ratio is given by \eq{ratioRandom} in the main text.
\end{proposition}
\begin{proof}
Fix a $k$-regular graph with $N$ vertices and, for $0<n<N$, let $\mathbf{u}_{n}$ denote the uniform distribution on the set of configurations, $\xi$, with exactly $n$ cooperators. Since $\mathbf{u}_{n}$ is independent of the graph geometry,
\begin{linenomath}
\begin{align}
\mathbf{u}_{n}\left[\xi\left(x\right)\widehat{\xi}\left(y\right)\right] &= \frac{n\left(N-n\right)}{N\left(N-1\right)}
\end{align}
\end{linenomath}
whenever $x\neq y$. Therefore, by the definitions of $f_{i}$ and $f_{ij}$ in \eq{localFrequencyDefinitions},
\begin{linenomath}
\begin{subequations}\label{eq:uniformMeans}
\begin{align}
\mathbf{u}_{n}\Big[\overline{f_{10}}\left(\xi\right)\Big] &= \frac{n\left(N-n\right)}{N\left(N-1\right)} ; \\
\mathbf{u}_{n}\Big[\overline{f_{1}f_{0}}\left(\xi\right)\Big] &= \frac{\left(k-1\right) n\left(N-n\right)}{kN\left(N-1\right)} .
\end{align}
\end{subequations}
\end{linenomath}
It follows from \eq{ratioGeneral} in the main text that the critical benefit-to-cost ratio for $\mathbf{u}_{n}$ is
\begin{linenomath}
\begin{align}
\left(\frac{b}{c}\right)_{\mathbf{u}_{n}}^{\ast} &= \frac{\mathbf{u}_{n}\left[k\left(N\overline{f_{1}}\left(\xi\right)\overline{f_{0}}\left(\xi\right) -\overline{f_{10}}\left(\xi\right)\right)\right]}{\mathbf{u}_{n}\left[N\overline{f_{1}}\left(\xi\right)\overline{f_{0}}\left(\xi\right) -k\overline{f_{10}}\left(\xi\right) -k\overline{f_{1}f_{0}}\left(\xi\right)\right]} = \frac{k\left(N-2\right)}{N-2k} ,
\end{align}
\end{linenomath}
which is independent of $n$ and coincides with \eq{ratioRandom}. Furthermore, one can use \eq{uniformMeans} to see that the coefficient of $w$ on the right-hand side of \eq{dbTheoremEquation} under the random placement $\mathbf{u}_{n}$ is equal to
\begin{linenomath}
\begin{align}
\frac{n\left(N-n\right)}{2N\left(N-1\right)}\Big[ b\left(N-2k\right) - ck\left(N-2\right) \Big] ,
\end{align}
\end{linenomath}
which is consistent with Theorem 1 in \citep{chen:AAP:2013}.
\end{proof}

\begin{remark}[Neutrality of random configurations]
On an arbitrary finite, connected social network, there is still an expansion in $w$ for fixation probabilities that generalizes \eq{firstOrderExpansion} \citep[see][Theorem 3.8]{chen:AAP:2013}. Moreover, for the donation game and a configuration given randomly by $\mathbf{u}_{n}$, this expansion takes the form
\begin{linenomath}
\begin{align}\label{eq:inefficiencyEquation}
\rho_{\mathbf{u}_{n},\mathbf{C}}\left(w\right) &= \rho_{\mathbf{u}_{n},\mathbf{C}}\left(0\right) + w\frac{n\left(N-n\right)}{N\left(N-1\right)}\left(b\Gamma_{1}-c\Gamma_{2}\right) + O\left(w^{2}\right) ,
\end{align}
\end{linenomath}
where $\Gamma_{1}$ and $\Gamma_{2}$ are constants that are independent of $n$, $b$, and $c$. (See the proof of Lemma 3.1 and the discussion of `Bernoulli transforms' on p. 655-656 in \citep{chen:AAP:2013}. For the linearity of the coefficient of $w$ in $b$ and $c$, see also \citep{tarnita:JTB:2009}.) By \eq{inefficiencyEquation}, the benefit-to-cost ratio for any $n$-random configuration is independent of $n$, so random configurations with more cooperators are neither more nor less conducive to cooperation than those with fewer.
\end{remark}

For a fixed graph, $G$, let $N_{0}$ be the maximum number of vertices that can be chosen in such a way that no two of these vertices are within two steps of one another. We say that a subset of vertices with this property is isolated. If the defectors in a configuration lie on isolated vertices, then we say that defectors are isolated.

\begin{proposition}\label{prop:cooperatorNumber}
If $N>2k$, then cooperation can be favored for a mixed configuration with $n$ cooperators whenever either $1\leqslant n\leqslant N_{0}+1$ or $1\leqslant N-n\leqslant N_{0}+1$.
\end{proposition}
\begin{proof}
For any configuration, $\xi$, with $n$ cooperators, we have the inequalities
\begin{linenomath}
\begin{subequations}\label{eq:fijInequalities}
\begin{align}
\overline{f_{10}}\left(\xi\right) &\leqslant \frac{n}{N} ; \label{eq:fijInequalityOne} \\
\overline{f_{1}f_{0}}\left(\xi\right) &\leqslant \frac{n\left(k-1\right)}{Nk} . \label{eq:fijInequalityTwo}
\end{align}
\end{subequations}
\end{linenomath}
In \eq{fijInequalities}, equality is obtained by a configuration with $n$ isolated cooperators. Indeed, $\overline{f_{10}}\left(\xi\right)$ and $\overline{f_{1}f_{0}}\left(\xi\right)$ depend on the number of cooperator-defector paths and the number of cooperator-anything-defector paths in $\xi$, respectively, and each such path is defined by either an edge or two incident edges. On the other hand, at least one of these inequalities is strict whenever $\xi$ does not have isolated cooperators: If two cooperators are adjacent to one another, then \eq{fijInequalityOne} is strict; if two cooperators are adjacent to the same defector, then \eq{fijInequalityTwo} is strict. In order to establish the proposition, we need to show that
\begin{linenomath}
\begin{align}\label{eq:aPriori}
N\overline{f_{1}}\left(\xi\right)\overline{f_{0}}\left(\xi\right) -k\overline{f_{10}}\left(\xi\right) -k\overline{f_{1}f_{0}}\left(\xi\right) > 0
\end{align}
\end{linenomath}
since, then, the critical benefit-to-cost ratio of \eq{ratioGeneral} is finite. Moreover, since \eq{aPriori} is invariant under conjugation, it suffices to consider configurations with $n$ cooperators, where $1\leqslant n\leqslant N_{0}+1$. By \eq{fijInequalities},
\begin{linenomath}
\begin{align}
N\overline{f_{1}}\left(\xi\right)\overline{f_{0}}\left(\xi\right) -k\overline{f_{10}}\left(\xi\right) -k\overline{f_{1}f_{0}}\left(\xi\right) &\geqslant \frac{n\left(N-n-2k+1\right)}{N} ,
\end{align}
\end{linenomath}
so it suffices to establish the inequality $N-N_{0}-2k\geqslant 0$.

Suppose, on the other hand, that $N-N_{0}-2k<0$. By the definition of $N_{0}$, we can then find a configuration with $N-2k+1$ isolated cooperators. Since each of these cooperators has $k$ neighboring defectors, and since none of these defectors have more than one cooperator as a neighbor, we have
\begin{linenomath}
\begin{align}
\left(N-2k+1\right)\left(k+1\right) \leqslant N \iff N \leqslant 2k ,
\end{align}
\end{linenomath}
which contradicts the assumption that $N>2k$, as desired.
\end{proof}

\begin{remark}
The proof of Proposition \ref{prop:cooperatorNumber} shows that whenever $N>2k$, in fact $N-2k\geqslant\max_{G}N_{0}$ holds, where $G$ ranges over all $k$-regular graphs on $N$ vertices. This lower bound, $\max_{G}N_{0}$, is sharp, which can be seen from the graph in \fig{promoteSuppress}(b) since this graph has size $9$, is $4$-regular, and satisfies $N_{0}=1$.
\end{remark}

\begin{proposition}\label{prop:reduceBC}
Suppose that $N>2k$. Let $\xi$ and $\xi '$ be configurations with $n$ and $n-1$ cooperators, respectively, such that defectors under both configurations are isolated. Then,
\begin{linenomath}
\begin{align}
\left(\frac{b}{c}\right)_{\xi}^{\ast} < \left(\frac{b}{c}\right)_{\xi '}^{\ast} .
\end{align}
\end{linenomath}
\end{proposition}
\begin{proof}
Since the defectors in both $\xi$ and $\xi '$ are isolated, we have
\begin{linenomath}
\begin{subequations}
\begin{align}
\left(\frac{b}{c}\right)_{\xi}^{\ast} &= \frac{k\left(n-1\right)}{n-2k+1} ; \\
\left(\frac{b}{c}\right)_{\xi '}^{\ast} &= \frac{k\left(n-2\right)}{n-2k} ,
\end{align}
\end{subequations}
\end{linenomath}
and it follows at once that $\left(\frac{b}{c}\right)_{\xi}^{\ast} < \left(\frac{b}{c}\right)_{\xi '}^{\ast}$ since $k\geqslant 2$, as desired.
\end{proof}

As a consequence of Proposition \ref{prop:reduceBC}, we see that among the configurations with an isolated strategy (cooperators or defectors), the minimum critical benefit-to-cost ratio is attained by any configuration with just a single cooperator.

\begin{proposition}
For a $k$-regular graph with $N>2k$, we have the following:
\begin{enumerate}

\item[(i)] if $N_{0}\geqslant 2$, then, for any $n$ with $2\leqslant n\leqslant N_{0}$, there exists a configuration with $n$ cooperators whose critical benefit-to-cost ratio is smaller than that of a random configuration;

\item[(ii)] for any configuration with exactly two cooperators, such that, furthermore, these two cooperators are neighbors, cooperation can be favored by weak selection. Moreover, the critical benefit-to-cost ratio for this configuration is smaller than that of a random configuration.

\end{enumerate}
\end{proposition}
\begin{proof}
By Proposition \ref{prop:reduceBC}, the critical benefit-to-cost ratio for any configuration with $n\geqslant 2$ isolated cooperators is greater than the critical benefit-to-cost ratio for any configuration with just a single cooperator. Recall now that this ratio for one cooperator is the same as the ratio for $n$ randomly-placed cooperators \citep{chen:AAP:2013}. By \eq{dbProofEquation}, the critical benefit-to-cost ratio for $\xi$ is of the form $\left(\frac{b}{c}\right)_{\xi}^{\ast}=\frac{\mathbf{N}_{\xi}}{\mathbf{D}_{\xi}}$ for some voter-model expectations, $\mathbf{N}_{\xi}$ and $\mathbf{D}_{\xi}$. Therefore, by a simple averaging argument, we see that for each $n$ with $2\leqslant n\leqslant N_{0}$, there must exist a configuration with $n$ cooperators whose critical ratio is smaller than that of a random configuration, \eq{ratioRandom}, which completes the proof of part (i) of the proposition.

Let $\xi$ be a configuration with two cooperators placed at adjacent vertices, $x$ and $y$. If $T\left(x,y\right)$ is the number of vertices adjacent to both $x$ and $y$, then straightforward calculations give
\begin{linenomath}
\begin{subequations}
\begin{align}
\overline{f_{10}}\left(\xi\right) &= \frac{2k-2}{Nk} ; \\
\overline{f_{1}f_{0}}\left(\xi\right) &= \frac{2k\left(k-1\right) -2T\left(x,y\right)}{Nk^{2}} .
\end{align}
\end{subequations}
\end{linenomath}
It then follows from the definition of the critical benefit-to-cost ratio that
\begin{linenomath}
\begin{align}
\left(\frac{b}{c}\right)_{\xi}^{\ast} &= \frac{k\left(N-3-\frac{1}{k}\right)}{N-2k+\frac{1}{k}T\left(x,y\right)} ,
\end{align}
\end{linenomath}
which is smaller than the ratio for random placement, \eq{ratioRandom}, because $N>2k$, which gives (ii).
\end{proof}

\subsection{Examples}

In Figs. \ref{si:fig:one} and \ref{si:fig:two}, we give examples of the relationship between the configuration and the critical benefit-to-cost ratio on three small graphs.

\begin{figure}[t]
\centering
\includegraphics[scale=1.0]{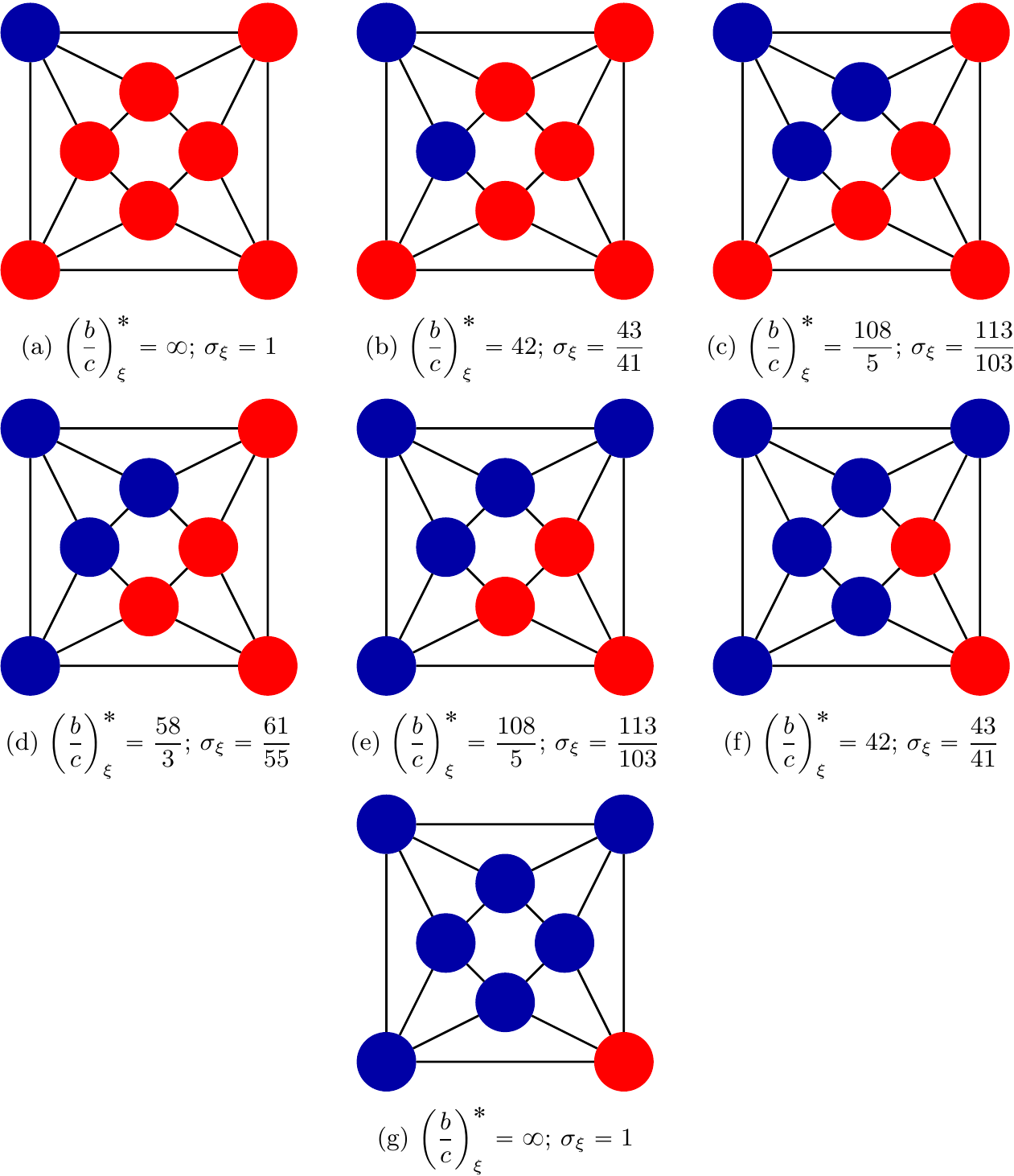}
\caption{Cooperator-defector configurations on a $4$-regular graph with $8$ vertices and diameter $2$. Starting from one cooperator in (a), a single cooperator is added in each subsequent panel. Although cooperation can never be favored by selection when starting from a state with a single mutant (a) or a single defector (g), it can be favored in the other states, (b)-(f), since the critical benefit-to-cost ratios are all finite in those panels.\label{si:fig:one}}
\end{figure}

%%%%

\begin{figure}[t]
\centering
\includegraphics[scale=1.0]{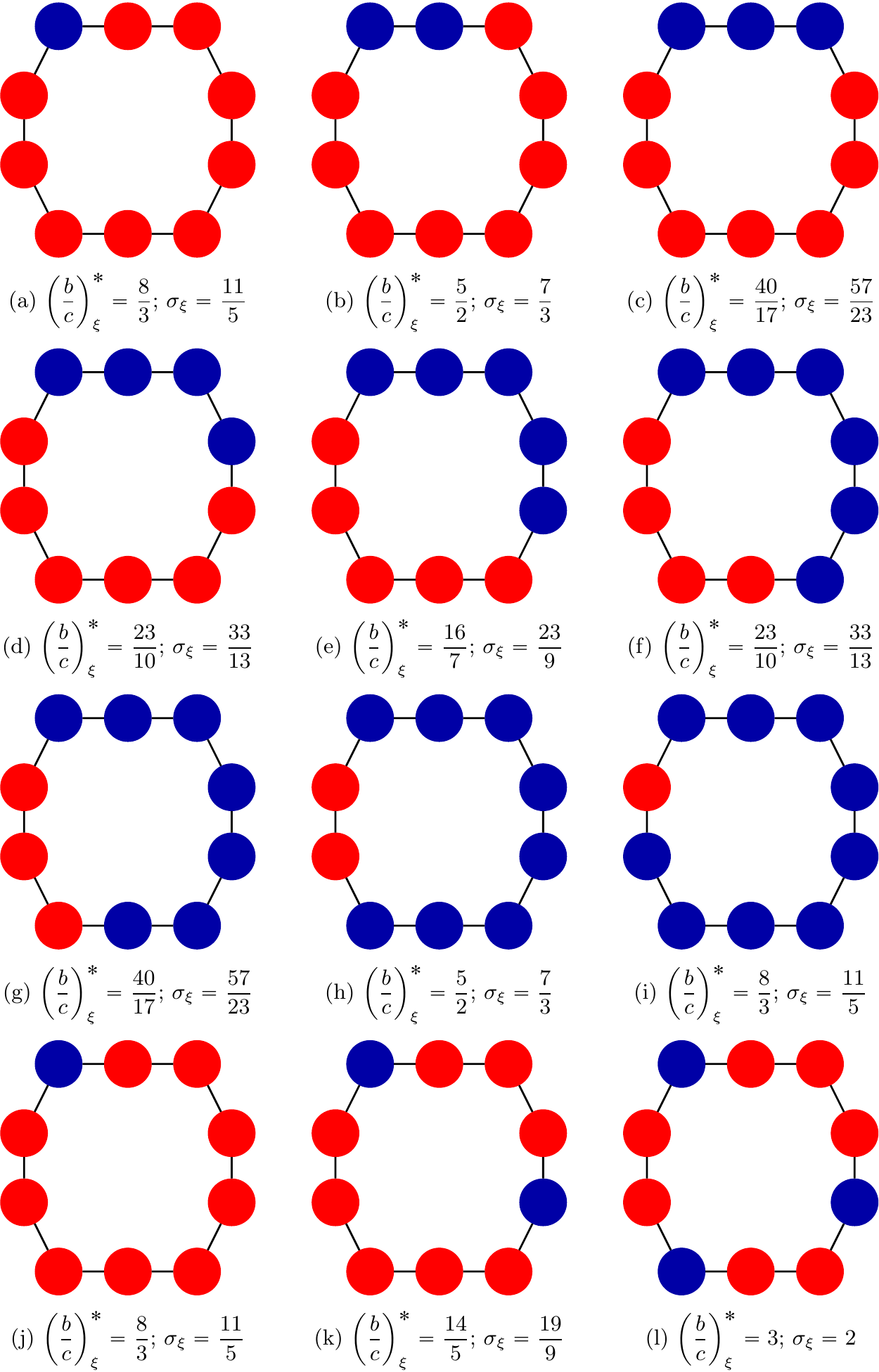}
\caption{The effects of adding cooperators to the initial condition on a cycle with $10$ vertices. In panels (a)-(i), cooperators are added sequentially, with each new cooperator neighboring a cooperator in the previous configuration. These panels clearly demonstrate that a configuration and its conjugate have the same critical ratio and structure coefficient. Panels (j)-(l) show that when cooperators are added in a different order, the critical ratios can increase rather than decrease. The configurations of (j)-(l) each have isolated cooperators.\label{si:fig:two}}
\end{figure}

\section*{Acknowledgments}
Support from the Center of Mathematical Sciences at Harvard University (Y.-T.C.), the John Templeton Foundation (Y.-T.C. and M.A.N.), a grant from B Wu and Eric Larson (Y.-T.C. and M.A.N.), and the Natural Sciences and Engineering Research Council of Canada (A.M.), and the Office of Naval Research (A.M. and M.A.N.) is gratefully acknowledged. Part of this research was done while Y.-T.C. visited NCTS Taipei.

\end{document}